%% file: manuscript.tex
\documentclass{elsarticle}
\usepackage[english]{babel}
\usepackage[utf8]{inputenc}
\usepackage[numbers]{natbib}

\usepackage[hidelinks]{hyperref}

\input{preamble}

\journal{Discrete Applied Mathematics SI:LAGOS 2019}

\begin{document}

\begin{frontmatter}

\title{Revising Johnson's table for the 21st century} 

\author[ufrj]{Celina M. H. de Figueiredo}
\ead{celina@cos.ufrj.br}

\author[ufrj]{Alexsander A. de Melo}
\ead{aamelo@cos.ufrj.br}

\author[uerj]{Diana Sasaki}
\ead{diana.sasaki@ime.uerj.br}

\author[ufc]{Ana Silva}
\ead{anasilva@mat.ufc.br}

\address[ufrj]{Federal University of Rio de Janeiro, Rio de Janeiro, Brazil}
\address[uerj]{Rio de Janeiro State University, Rio de Janeiro, Brazil}
\address[ufc]{Federal University of Cear\'{a}, Cear\'{a}, Brazil}

\input{00-Abstract}

\end{frontmatter}


\input{01-Introduction}

\input{02-Parameterized}

\input{03-Reductions}

\input{04-Steiner}

\input{05-Conclusion}

\section*{Acknowledgements}
We would like to thank anonymous reviewers for their thorough reading and numerous insightful suggestions and comments; in particular, we would like to thank one of the reviewers for the references related to mim-width. 
We are grateful to Vinicius F. Santos for sharing the reference~\cite{ABMR.20}, for the \NP-completeness proof of \problemName{MaxCut} for \className{Interval} graphs. 
This work was partially supported by the Brazilian agencies CNPq (Grant numbers: 407635/2018-1, 140399/2017-8, 407430/2016-4, 437841/2018-9, 401519/2016-3, and 304576/2017-4), CAPES (Finance Code 001), FAPERJ (Grant number: CNE E-26/202.793/2017), and CNPq/FUNCAP (PNE-0112-00061.01.00/16).

\bibliography{manuscript}

\end{document}

%% file: preamble.tex
\usepackage{amsmath,amssymb,mathtools}
\usepackage{amsthm}
\usepackage{bm}
\usepackage{environ}
\usepackage[small]{complexity}
\usepackage{tabularx,multirow}
\usepackage{times,graphicx,subfigure}
\usepackage{graphics}
\usepackage{setspace}
\usepackage{xspace}

\usepackage{booktabs}
\usepackage{pdflscape}
\usepackage[dvipsnames]{xcolor}
\usepackage{url}
\usepackage{tikz}
\usetikzlibrary{positioning,fit,calc,shapes,backgrounds}
\usetikzlibrary{decorations.pathmorphing}
\usetikzlibrary{matrix,graphs}

\definecolor{blue(ncs)}{rgb}{0.0, 0.53, 0.74}
\definecolor{darkolivegreen}{rgb}{0.33, 0.42, 0.18}
\definecolor{armygreen}{rgb}{0.29, 0.33, 0.13}
\definecolor{tomato}{rgb}{1.0, 0.39, 0.28}
\definecolor{carnelian}{rgb}{0.7, 0.11, 0.11}
\definecolor{calpolypomonagreen}{rgb}{0.12, 0.3, 0.17}
\definecolor{darkmagenta}{rgb}{0.55, 0.0, 0.65}
\definecolor{red-violet}{rgb}{0.80, 0.08, 0.42}
\definecolor{navyblue}{rgb}{0.0, 0.0, 0.5}
\definecolor{dollarbill}{rgb}{0.52, 0.73, 0.4}
\definecolor{islamicgreen}{rgb}{0.0, 0.56, 0.0}
\definecolor{darkpastelgreen}{rgb}{0.01, 0.75, 0.24}
\definecolor{sapgreen}{rgb}{0.31, 0.49, 0.16}
\definecolor{emerald}{rgb}{0.31, 0.78, 0.47}

\DeclareMathAlphabet{\mathsfit}{T1}{\sfdefault}{\mddefault}{\sldefault}
\SetMathAlphabet{\mathsfit}{bold}{T1}{\sfdefault}{\bfdefault}{\sldefault}

\newclass{\COpen}{{O}}
\newclass{\CGI}{{I}}
\newclass{\CNPC}{{N}}

\newclass{\CWONE}{\ensuremath{\W_{1}}}
\newclass{\CWTWO}{\ensuremath{\W_{2}}}

\newclass{\CparaNP}{{pN}}
\newclass{\paraNP}{{paraNP}}

\newcommand{\Open}{\textcolor{islamicgreen}{\COpen}}
\newcommand{\EOpen}{\textcolor{emerald}{\COpen?}}
\newcommand{\Poly}{\textcolor{blue(ncs)}{\P}}

\newcommand{\NPh}{\textcolor{tomato}{\CNPC}}
\newcommand{\ONPh}{\textcolor{tomato}{\COpen*}}
\newcommand{\aGI}{\textcolor{Brown}{\CGI}}
\newcommand{\PNP}{\textcolor{BrickRed}{\CparaNP}}
\newcommand{\WONE}{\textcolor{darkmagenta}{\CWONE}}
\newcommand{\WTWO}{\textcolor{red-violet}{\CWTWO}}
\newcommand{\TXP}{\textcolor{sapgreen}{\XP}}
\newcommand{\TFPT}{\textcolor{navyblue}{\FPT}}

\newcommand{\BPoly}{\ensuremath{\bm{\textcolor{blue}{\P}}}}
\newcommand{\BNPh}{\ensuremath{\bm{\textcolor{red}{\CNPC}}}}
\newcommand{\IBNPh}{\ensuremath{\bm{\mathsfit{\textcolor{red}{N}}}}}

\newcommand{\BGI}{\ensuremath{\bm{\aGI}}}

\newcommand{\BTXP}{\ensuremath{\bm{\BTXP}}}
\newcommand{\BTFPT}{\ensuremath{\bm{\BTFPT}}}

\newcommand{\pc}[1]{\underline{#1}}

\DeclarePairedDelimiter{\set}{\lbrace}{\rbrace}
\DeclarePairedDelimiter{\abs}{|}{|}
\DeclarePairedDelimiter{\tuple}{(}{)}
\DeclarePairedDelimiter{\sequence}{\langle}{\rangle}
\newcommand{\setst}{\colon}
\newcommand{\BigOh}[1]{\mathcal{O}(#1)}

\newtheorem{theorem}{Theorem}
\newtheorem{lemma}[theorem]{Lemma}
\newtheorem{proposition}[theorem]{Proposition}
\newtheorem{corollary}[theorem]{Corollary}

\newcommand{\problemName}[1]{\textsc{#1}}
\newcommand{\className}[1]{\textsc{#1}}
\newcommand{\no}{\texttt{no}}
\newcommand{\yes}{\texttt{yes}}

\makeatletter
\newcommand{\problemtitle}[1]{\gdef\@problemtitle{#1}}
\newcommand{\probleminput}[1]{\gdef\@probleminput{#1}}
\newcommand{\problemquestion}[1]{\gdef\@problemquestion{#1}}
\NewEnviron{myproblem}{
  \problemtitle{}\probleminput{}\problemquestion{}
  \BODY
  \par\addvspace{.5\baselineskip}
  \noindent
  \begin{tabularx}{\textwidth}{@{\hspace{\parindent}}l X c}
    \multicolumn{2}{@{\hspace{\parindent}}l}{\@problemtitle} \\
    \textbf{Input:} & \@probleminput \\
    \textbf{Question:} & \@problemquestion
  \end{tabularx}
  \par\addvspace{.5\baselineskip}
}

\newcommand{\domset}{\ensuremath{D}}
\newcommand{\matching}{\ensuremath{M}}

\newcommand{\instance}{\ensuremath{I}}

\newcommand{\param}{\ensuremath{\kappa}}

\newcommand{\graph}{\ensuremath{G}}
\newcommand{\tree}{\ensuremath{T}}
\newcommand{\cliquetree}{{T}}

\newcommand{\vset}{\ensuremath{V}}
\newcommand{\eset}{\ensuremath{E}}
\newcommand{\gvset}[1]{\vset(#1)}
\newcommand{\geset}[1]{\eset(#1)}

\newcommand{\nvertices}{\ensuremath{n}}
\newcommand{\nedges}{\ensuremath{m}}

\newcommand{\uvertex}{\ensuremath{u}}
\newcommand{\vvertex}{\ensuremath{v}}

\newcommand{\avertex}{\ensuremath{a}}
\newcommand{\bvertex}{\ensuremath{b}}
\newcommand{\cvertex}{\ensuremath{c}}
\newcommand{\xvertex}{\ensuremath{x}}
\newcommand{\yvertex}{\ensuremath{y}}
\newcommand{\zvertex}{\ensuremath{z}}
\newcommand{\pvertex}{\ensuremath{p}}
\newcommand{\qvertex}{\ensuremath{q}}
\newcommand{\rvertex}{\ensuremath{r}}

\newcommand{\striple}{\mathbf{s}}

\newcommand{\uset}{\ensuremath{U}}
\newcommand{\pset}{\ensuremath{P}}
\newcommand{\qset}{\ensuremath{Q}}
\newcommand{\rset}{\ensuremath{R}}
\newcommand{\sset}{\mathcal{S}}
\newcommand{\kclique}{\ensuremath{K}}

\newcommand{\cnh}[2]{\ensuremath{{N}_{#1}[#2]}}

\newcommand{\alexsander}[1]{{\color{blue} Alexsander: #1}}
\newif\ifnotes

%% file: 00-Abstract.tex
\begin{abstract}
What does it mean today to study a problem from a computational point of view? We focus on parameterized complexity and on Column 16 ``Graph Restrictions and Their Effect” of D. S. Johnson's Ongoing guide, where several puzzles were proposed in a summary table with 30 graph classes as rows and 11 problems as columns. 
Several  of the 330 entries remain unclassified into Polynomial or \NP-complete after 35 years. We provide a full dichotomy for the \problemName{Steiner Tree} column by proving that the problem is \NP-complete when restricted to \className{Undirected Path} graphs. 
We revise Johnson's summary table according to the granularity provided by the parameterized complexity for \NP-complete problems.
\end{abstract}

\begin{keyword}
computational complexity, parameterized complexity, \NP-complete, Steiner tree, dominating set
\end{keyword}

%% file: 01-Introduction.tex
\section{Graph Restrictions and Their Effect 35 Years Later}%
The 1979 book \emph{Computers and Intractability, A Guide to the Theory of \NP-complete\-ness} by Michael R. Garey and David S. Johnson~\cite{GJ.79} is considered the single most important book by the computational complexity community and it is known as The Guide, which we cite by [GJ]. 
The book was followed by \emph{The \NP-completeness Column: An Ongoing Guide} where, from 1981 until 2007, D. S. Johnson continuously updated The Guide in 26 columns published first in the \emph{Journal of Algorithms} and then in the \emph{ACM Transactions on Algorithms}. 
The Guide has an appendix where 300 \NP-complete problems are organized into 13 categories according to subject matter. The first, ``A1 Graph Theory'', contains 65 problems and the second, ``A2 Network Design'', contains 51 problems. 
Category ``A13 Open Problems'' lists 12 problems in \NP, at the time not classified into polynomial or \NP-complete, and it is surprising that since then 5 have been classified into polynomial and 5 into \NP-complete.
Garey and Johnson were amazingly able to foresee a list of challenging problems which would evenly split into tractable and intractable. 
The goal of the present paper is to propose an answer to the question: \emph{What does it mean today to study a problem from a computational complexity point of view?}
In search of an answer, we focus on  parameterized complexity and on Column 16 ``Graph Restrictions and Their Effect''~\cite{J.85}, which we cite by [OG],
where several puzzles were proposed by D. S. Johnson and many remain unsolved after 35 years. 
Consider in Table~\ref{table:first} the summary table from [OG] with 30 graph classes as rows and 11 columns, the first of which is \problemName{Membership}, followed by 10 well-known \NP-complete problems, listed in The Guide's appendix as GT20, GT19, GT15, GT13, Open5, GT37, GT2, ND16, ND12, and Open1.
The entries follow the notation of [OG], where the complexity of the problem restricted to the graph class is {\CNPC} = \NP-complete, {\P} = polynomial, or {\COpen} = open. 
Following our convention,
reference [GJ] stands for ``The Guide'' and [OG] stands for ``Column 16'', and we highlight in bold the reference updates with the corresponding new recent references. 
Every reference associated with each entry of Table~\ref{table:first} has been checked, and the updated entries are precisely those that needed to be updated.
It is surprising that several {\COpen}? entries remain stubbornly open. At the time, D. S. Johnson proposed only one ``{\COpen}! = famous open problem'', \problemName{Membership} for \className{Perfect} graphs, which we know today to be in  \P, and two entries ``{\COpen} = may well be hard'', \problemName{Hamiltonian Circuit} restricted to \className{Permutation} graphs, known today to be in \P, and \problemName{Chromatic Index} restricted to \className{Planar} graphs, which is still open. 
We remark that in the original summary table [OG], there was only one entry co-authored by a Brazilian researcher among 330 entries, namely \problemName{Hamiltonian Circuit} restricted to \className{Grids}~\cite{IPS.82}, and today we have two additional such entries: \problemName{Maximum Cut} restricted to \className{Strongly Chordal}~\cite{SFK.13} and \problemName{Graph Isomorphism} restricted to \className{Proper Circular Arc}~\cite{LSS.08}.

\input{table}

We depict in Figures~\ref{fig:classesfirst} and~\ref{fig:classessecond}  the relations between the graph classes, and use the convention from [OG] that an arrow from \textsc{Class A} to \textsc{Class B} means that \textsc{Class A} contains \textsc{Class B}. 
Since the only {\COpen}! entry was  \problemName{Membership} for \className{Perfect} graphs, the chosen 30 classes were classified into the following four categories: Trees and Near-Trees, Planarity and its Relations, A Catalog of Perfect Graphs, and Intersection Graphs.
Although very similar to the figures presented in [OG], our Figures~\ref{fig:classesfirst} and~\ref{fig:classessecond} present some additional relations, that were either unknown or unobserved, and about which we comment next.

Figure~\ref{fig:classesfirst} highlights the key property that 7 graph classes are subclasses of \className{Partial $k$-Trees}~\cite{B.98}, also known as \className{Bounded Treewidth} graphs.
Also, although Table~\ref{table:first} follows the same organization as the one used in [OG], our proposed new Table~\ref{table:second} is organized in a way as to highlight this relationship, with the first~8 rows being exactly \className{Partial $k$-Trees} and its subclasses. 
In the summary table, D.S. Johnson used entry ``\P? =  appears to be polynomial-time solvable by standard techniques, but I haven’t checked the details''. D.S. Johnson is correct, since all \P? entries are known today to be {\P} entries. All former \P? entries used to
appear in these 8 rows, \className{Partial $k$-Trees} and the 7 subclasses~\cite{B.98}.

Another difference is that we use \className{$K_{3,3}$-Free\textsuperscript{*}} to denote the graph class referred to in [OG] by \className{$K_{3,3}$-Free}. This is to avoid confusion, since nowadays it is standard to use \className{$H$-Free} to refer to the class of graphs that do not contain $H$ as an \emph{induced subgraph}. However, the class investigated in [OG] is instead the class of graphs that contain no subgraphs \emph{homeomorphic} to $K_{3,3}$; in other words, the class of graphs that do not contain $K_{3,3}$ as a topological minor. Observe that, using our notation, we have that \className{$K_{3,3}$-Free\textsuperscript{*}} is a proper subclass of \className{$K_{3,3}$-Free}. 
Also, we mention that this confusion does not occur for \className{Claw-Free} graphs, since we, as well as [OG], use it to denote the class of graphs that do not contain $K_{1,3}$ (also known as the claw) as an induced subgraph.

Finally, we mention two relations that do not appear in [OG], both involving the class \className{Thickness-$k$}. A graph $G$ is said to have \emph{thickness at most} $k$ if $E(G)$ can be partitioned into at most $k$ subsets, each of which forms a planar subgraph of $G$. 
In the same way as all the other graph classes that have a parameter in their names such as the class of \className{Partial $k$-Trees} that is also known as \className{Bounded Treewidth} graphs,
the class \className{Thickness-$k$} means 
\className{Bounded Thickness} graphs.
First, note that if $G$ has degree at most $\Delta(G)$, then by Vizing's Theorem, we get that $E(G)$ can be colored with at most $\Delta(G)+1$ colors. In other words, this means that the edge set of $G$ can be partitioned into $\Delta(G)+1$ matchings, which are planar graphs, and hence $G$ has thickness at most $\Delta(G)+1$. Therefore, we get that \className{Degree-$k$} is a subclass of \className{Thickness-$k$}. 
Another non-trivial relation involving this class is with \className{Genus-$k$} graphs. 
A \emph{$k$-book embedding} of a graph $G$ is a linear ordering of its vertices along the spine of a book and an assignment of its edges to $k$ pages so that edges assigned to the same page can be drawn on that page without crossings.
The \emph{pagenumber} of a graph $G$ is the minimum $k$ for which $G$ admits a $k$-book embedding. 
Clearly, the pagenumber of $G$ is an upper bound for the thickness of $G$. Additionally, in~\cite{Malitz1994} the authors prove that the pagenumber of $G$ is bounded above by a function of the genus of $G$. This means that if $G$ has bounded genus, then $G$ also has bounded thickness. Therefore, \className{Genus-$k$} is a subclass of \className{Thickness-$k$}. 
On the other hand, to the best of our knowledge, it is not known whether graphs with bounded thickness have bounded genus. 

\input{figdiagrama_partialktrees}
\input{figdiagrama_perfect}

\paragraph{Our contribution}In the summary table, the \problemName{Steiner Tree} column had 6 unresolved entries: 5 \P? entries, all of which are now known to be subclasses of \className{Partial $k$-Trees} and henceforth are in $\P$, and one {\COpen}? entry for \className{Undirected Path} graphs. Upon close investigation of the references given in [GJ] and [OG], we found that many consist of ``private communication'' or could not be confirmed.
In the particular case of the \problemName{Steiner Tree} column, we found that this happens for the lines \className{Circular Arc}, \className{Circle}, \className{Proper Circular Arc},
\className{Edge (or Line)}, and  \className{Claw-Free}.
 We were able to find a recent reference for \className{Edge (or Line)} and \className{Claw-Free}.
 Additionally, based on the facts that \className{Circular Arc} and, consequently, \className{Proper Circular Arc} graphs have bounded mim-width and that {\sc Steiner tree} is polynomial-time solvable for graphs of bounded mim-width, we were able to resolve such entries as well. 
However, we could not find any reference for \className{Circle} graphs, and therefore we underline the corresponding [OG] reference in Table~\ref{table:first}.  
 Moreover, the entry \className{Undirected Path} is said to be $\NP$-complete in~\cite{S.book}, but again with a ``private communication'' reference (we comment more on this in Section~\ref{sec:ST}). Believing in the need to have explicit proofs for these important problems, we here give a proof of $\NP$-completeness for \className{Undirected Path} graphs, which would provide a full dichotomy Polynomial versus \NP-complete for the \problemName{Steiner Tree} column.
Actually, we provide a second dichotomy for the \problemName{Steiner Tree} problem restricted to \className{Undirected Path} graphs, according to the diameter of the input graph.
For the \problemName{Graph Isomorphism} column we also provide 
a full dichotomy Polynomial versus \NP-complete by giving an explicit proof of \GI-completeness for \className{Thickness-$k$} graphs (please refer to Section~\ref{sec:simplered}).


Besides providing a full dichotomy Polynomial versus \NP-complete for the \problemName{Steiner Tree} column, in Table~\ref{table:first} we have thoroughly revised the summary table that 35 years later has 54 new resolved entries depicted in bold.
Additionally, there are 36 citations for references not in bold that confirm resolved entries from [OG] or [GJ], that we updated because they cited private communications, or because the cited reference is not easily accessible, or could not be confirmed.
There is one entry highlighted in italic that corrects the entry for \problemName{Hamiltonian Circuit} restricted to \className{Circle} graphs originally $\P$ but that actually is {\CNPC}~\cite{D.89}.

In addition, we consider the parameterized complexity of hard problems to revise Table~\ref{table:first} into a new Table~\ref{table:second}, a proposed summary table of what it means today to study a problem from a computational complexity point of view. This is of course just a sample of what it means, since we could even consider other classifications (e.g., the approximability complexity theory and the space complexity theory). We have kept the same 30 classes but have drawn the horizontal lines so that the \className{Partial $k$-Trees} subclasses appear together, and we may focus on the remaining rows, where the {\NP}-complete entries appear. In Section~\ref{sec:parameterized}, we discuss in detail Table~\ref{table:second}, also presenting the basic definitions of parameterized complexity, in order to draw the reader's attention to the granularity provided by the parameterized complexity for the $\NP$-complete problems into {\XP}, {\FPT}, {\CWONE}, {\CWTWO}, and {\CparaNP}. We depict in Table~\ref{table:second} as {\COpen*} the only {\CNPC} entry of Table~\ref{table:first} that constitutes the parameterized puzzle for which so far we were not able to provide a parameterized complexity classification. This is to show how rich the original problems posed by Garey and Johnson are, and how their initial classification continues to develop into ever evolving complexity classes, with the $\NP$-complete class being now just the beginning of a very interesting story.

%% file: table.tex
%
%


\begin{table}[ht]\centering\small
\renewcommand{\arraystretch}{1.2}
\setlength\tabcolsep{3.6pt}
\resizebox{\linewidth}{!}{
\begin{tabular}{@{}l|ll|llllllllllllllllllll@{}}
\toprule
\textsc{Graph Class} & \multicolumn{2}{l|}{\problemName{Member}} & \multicolumn{2}{l}{\problemName{IndSet}} & \multicolumn{2}{l}{\problemName{Clique}} & \multicolumn{2}{l}{\problemName{CliPar}} & \multicolumn{2}{l}{\problemName{ChrNum}} & \multicolumn{2}{l}{\problemName{ChrInd}} & \multicolumn{2}{l}{\problemName{HamCir}} & \multicolumn{2}{l}{\problemName{DomSet}} & \multicolumn{2}{l}{\problemName{MaxCut}} & \multicolumn{2}{l}{\problemName{StTree}} & \multicolumn{2}{l}{\problemName{GraphIso}} \\ \midrule

\className{Trees/Forests} & \Poly & [T] & \Poly & [GJ] & \Poly & [T] & \Poly & [GJ] & \Poly & [T] & \Poly & [GJ] & \Poly & [T] & \Poly & [GJ] & \Poly & [GJ] & \Poly & [T] & \Poly & [GJ] \\

\className{Almost Trees ($k$)} & \Poly & [OG] & \Poly & [OG] & \Poly & [T] & \BPoly & \textbf{\cite{blanchette2012}} & \BPoly & \textbf{\cite{APDAM89}} & \BPoly & \textbf{\cite{bodlaender1990}} & \BPoly & \textbf{\cite{APDAM89}} & \Poly & \cite{APDAM89} & \BPoly & \textbf{\cite{bodlaender1994}} & \BPoly & \textbf{\cite{korach1990}} & \BPoly & \textbf{\cite{bodlaender1990}} \\

\className{Partial $k$-trees} & \Poly & [OG] & \Poly & \cite{APDAM89} & \Poly & [T] & \BPoly & \textbf{\cite{blanchette2012}} & \Poly & \cite{APDAM89} & \BPoly & \textbf{\cite{bodlaender1990}}  & \Poly & \cite{APDAM89} & \Poly & \cite{APDAM89} & \BPoly & \textbf{\cite{bodlaender1994}} & \BPoly & \textbf{\cite{korach1990}} & \BPoly & \textbf{\cite{bodlaender1990}} \\

\className{Bandwidth-$k$} & \Poly & [OG] & \Poly & [OG] & \Poly & [T] & \BPoly & \textbf{\cite{blanchette2012}} & \Poly & \cite{APDAM89} & \BPoly & \textbf{\cite{bodlaender1990}} & \BPoly & \textbf{\cite{APDAM89}} & \Poly & \cite{APDAM89} & \Poly & [OG] & \BPoly & \textbf{\cite{korach1990}} & \Poly & [OG] \\

\className{Degree-$k$} & \Poly & [T] & \NPh & [GJ] & \Poly & [T] & \NPh & \cite{CERIOLI20082270} & \NPh & [GJ] & \NPh & [OG] & \NPh & [GJ] & \NPh & [GJ] & \NPh & [GJ] & \NPh & [GJ] & \Poly & [OG] \\

\midrule

\className{Planar} & \Poly & [GJ] & \NPh & [GJ] & \Poly & [T] & \NPh & \cite{Kral2001} & \NPh & [GJ] & \Open &  & \NPh & [GJ] & \NPh & [GJ] & \Poly & [GJ] & \NPh & [OG] & \Poly & [GJ] \\

\className{Series Parallel} & \Poly & [OG] & \Poly & [OG] & \Poly & [T] & \BPoly & \textbf{\cite{blanchette2012}} & \Poly & \cite{APDAM89} & \Poly & \cite{bodlaender1990} & \Poly & \cite{APDAM89} & \Poly & [OG] & \Poly & [GJ] & \Poly & [OG] & \Poly & [GJ] \\

\className{Outerplanar} & \Poly & [OG] & \Poly & [OG] & \Poly & [T] & \Poly & [OG] & \Poly & [OG] & \Poly & [OG] & \Poly & [T] & \Poly & [OG] & \Poly & [GJ] & \Poly & [OG] & \Poly & [GJ] \\

\className{Halin} & \Poly & [OG] & \Poly & [OG] & \Poly & [T] & \Poly & [OG] & \Poly & \cite{APDAM89} & \Poly & \cite{bodlaender1990} & \Poly & [T] & \Poly & [OG] & \Poly & [GJ] & \BPoly & \textbf{\cite{winter1987}} & \Poly & [GJ] \\

\className{$k$-Outerplanar} & \Poly & [OG] & \Poly & [OG] & \Poly & [T] & \Poly & [OG] & \Poly & \cite{APDAM89} & \BPoly & \textbf{\cite{bodlaender1990}} & \Poly & [OG] & \Poly & [OG] & \Poly & [GJ] & \BPoly & \textbf{\cite{korach1990}} & \Poly & [GJ] \\

\className{Grid} & \Poly & [OG] & \Poly & [GJ] & \Poly & [T] & \Poly & [GJ] & \Poly & [T] & \Poly & [GJ] & \NPh & [OG] & \NPh & \cite{Clark1990} & \Poly & [T] & \NPh & [OG] & \Poly & [GJ] \\

\className{$K_{3,3}$-Free\textsuperscript{*}} & \Poly & [OG] & \NPh & [GJ] & \Poly & [T] & \NPh & \cite{Kral2001} & \NPh & [GJ] & \EOpen &  & \NPh & [GJ] & \NPh & [GJ] & \Poly & [OG] & \NPh & [GJ] & \BPoly & \textbf{\cite{samir2009}} \\

\className{Thickness-$k$} & \NPh & [OG] & \NPh & [GJ] & \Poly & [T] & \NPh & \cite{Kral2001} & \NPh & [GJ] & \NPh & [OG] & \NPh & [GJ] & \NPh & [GJ] & \NPh & \cite{Y.78} & \NPh & [GJ] & \BGI & \textbf{Prop. \ref{prop:isoThickness}}  \\

\className{Genus-$k$} & \Poly & [OG] & \NPh & [GJ] & \Poly & [T] & \NPh & \cite{Kral2001} & \NPh & [GJ] & \EOpen &  & \NPh & [GJ] & \NPh & [GJ] & \EOpen &  & \NPh & [GJ] & \Poly & [OG] \\ \midrule

\className{Perfect} & \BPoly & \textbf{\cite{cornuejols2013}} & \Poly & [OG] & \Poly & [OG] & \Poly & [OG] & \Poly & [OG] & \BNPh & \textbf{\cite{cai1991}}  & \NPh & [OG] & \NPh & [OG] & \BNPh & \textbf{\cite{bodlaender1994}} & \NPh & [GJ] & \aGI & \cite{LB.79} \\

\className{Chordal} & \Poly & [OG] & \Poly & [OG] & \Poly & [OG] & \Poly & [OG] & \Poly & [OG] & \EOpen &  & \NPh & \cite{muller1996} & \NPh & [OG] & \BNPh & \textbf{\cite{bodlaender1994}} & \NPh & [OG] & \aGI & \cite{LB.79} \\

\className{Split} & \Poly & [OG] & \Poly & [OG] & \Poly & [OG] & \Poly & [OG] & \Poly & [OG] & \EOpen &  & \NPh & \cite{muller1996} & \NPh & [OG] & \BNPh & \textbf{\cite{bodlaender1994}} & \NPh & [OG] & \aGI & \cite{S.78} \\

\className{Strongly Chordal} & \Poly & [OG] & \Poly & [OG] & \Poly & [OG] & \Poly & [OG] & \Poly & [OG] & \EOpen &  & \BNPh & \textbf{\cite{muller1996}} & \Poly & [OG] & \BNPh & \textbf{\cite{SFK.13}} & \Poly & [OG] & \BGI & \textbf{\cite{uehara2005}}  \\

\className{Comparability} & \Poly & [OG] & \Poly & [OG] & \Poly & [OG] & \Poly & [OG] & \Poly & [OG] & \BNPh & \textbf{\cite{cai1991}} & \NPh & [OG] & \NPh & \cite{MB.87} & \BNPh & \textbf{\cite{pocai2016}} & \NPh & [GJ] & \aGI & \cite{booth1979b} \\

\className{Bipartite} & \Poly & [T] & \Poly & [GJ] & \Poly & [T] & \Poly & [GJ] & \Poly & [T] & \Poly & [GJ] & \NPh & [OG] & \NPh & \cite{MB.87} & \Poly & [T] & \NPh & [GJ] & \aGI & \cite{booth1979b} \\

\className{Permutation} & \Poly & [OG] & \Poly & [OG] & \Poly & [OG] & \Poly & [OG] & \Poly & [OG] & \EOpen &  & \BPoly & \textbf{\cite{deogun1994}}  & \Poly & [OG] & \EOpen &  & \Poly & [OG] & \Poly & [OG] \\

\className{Cographs} & \Poly & [T] & \Poly & [OG] & \Poly & [OG] & \Poly & [OG] & \Poly & [OG] & \EOpen &  & \Poly & [OG] & \Poly & [OG] & \BPoly & \textbf{\cite{bodlaender1994}} & \Poly & [OG] & \Poly & [OG] \\ \midrule

\className{Undirected} Path & \Poly & [OG] & \Poly & [OG] & \Poly & [OG] & \Poly & [OG] & \Poly & [OG] & \EOpen &  & \BNPh & \textbf{\cite{bertossi1986}} & \NPh & [OG] & \BNPh & \textbf{\cite{bodlaender1994}} & \BNPh & \textbf{Thm.~\ref{thm:SteinerUP}} & \aGI & \cite{booth1979b} \\

\className{Directed Path} & \Poly & [OG] & \Poly & [OG] & \Poly & [OG] & \Poly & [OG] & \Poly & [OG] & \EOpen &  & \BNPh & \textbf{\cite{panda2008np}}  & \Poly & [OG] & \BNPh & \textbf{\cite{ABMR.20}}  & \Poly & [OG] & \BPoly & \textbf{\cite{Babel1996}}  \\

\className{Interval} & \Poly & [OG] & \Poly & [OG] & \Poly & [OG] & \Poly & [OG] & \Poly & [OG] & \EOpen &  & \Poly & [OG] & \Poly & [OG] & \BNPh & \textbf{\cite{ABMR.20}} & \Poly & [OG] & \Poly & [OG] \\

\className{Circular Arc} & \Poly & [OG] & \Poly & [OG] & \Poly & [OG] & \Poly & [OG] & \NPh & [OG] & \EOpen &  & \BPoly & \textbf{\cite{shih1992n}} & \Poly & [OG] & \BNPh & \textbf{\cite{ABMR.20}} & \Poly & \cite{BK.19} & \BPoly & \textbf{\cite{Krawczyk2019}} \\

\className{Circle} & \Poly & [OG] & \Poly & [GJ] & \Poly & [OG] & \BNPh & \textbf{\cite{keil2006}} & \NPh & [OG] & \EOpen &  & \IBNPh & \textit{\textbf{\cite{D.89}}} & \BNPh & \textbf{\cite{K.93}}  & \BNPh & \textbf{\cite{buchheim2006}} & \Poly & [\pc{OG}] & \BPoly & \textbf{\cite{Kalisz2019}} \\

\className{Proper Circ. Arc} & \Poly & [OG] & \Poly & [OG] & \Poly & [OG] & \Poly & [OG] & \Poly & [OG] & \EOpen &  & \Poly & [OG] & \Poly & [OG] & \EOpen &  & \Poly & \cite{BK.19} & \BPoly & \textbf{\cite{LSS.08}}  \\

\className{Edge (or Line)} & \Poly & [OG] & \Poly & [GJ] & \Poly & [T] & \NPh & \cite{MUNARO20172208} & \NPh & [OG] & \BNPh & \textbf{\cite{cai1991}}  & \NPh & [OG] & \NPh & [GJ] & \BPoly & \textbf{\cite{guruswami1999}}  & \NPh & \cite{BBJPPL.arxiv} & \aGI & [OG] \\

\className{Claw-Free} & \Poly & [T] & \Poly & [OG] & \BNPh & \textbf{\cite{poljak1974}} & \NPh & \cite{MP.96} & \NPh & [OG] & \BNPh & \textbf{\cite{cai1991}}  & \NPh & [OG] & \NPh & [GJ] & \BNPh & \textbf{\cite{bodlaender1994}} & \NPh & \cite{BBJPPL.arxiv} & \aGI & [OG] \\ \bottomrule
\end{tabular}
}

\caption{The updated NP-Completeness Column: An Ongoing Guide table 35 years later. Depicted in bold are the references that correspond to unresolved entries in [OG] and [GJ]. The references not in bold confirm resolved entries from [OG] or [GJ], that we updated either because they cited private communications, because the cited reference is not easily accessible, or could not be confirmed.
There is one entry highlighted in italic that corrects the entry for \problemName{HamCirc} restricted to \className{Circle graphs}.
We keep the abbreviations used by [OG], namely for entries: {\P} = Polynomial-time solvable; {\CNPC} = \NP-complete; {\CGI} = Open, but equivalent in complexity to general \problemName{Graph Isomorphism}; {\COpen}? = Apparently open, but possibly easy to resolve; and {\COpen} = Open, and may well be hard; and for references [T] = Restriction trivializes the problem; [GJ] = the Guide~\cite{GJ.79}; and [OG] = the Ongoing guide~\cite{J.85}, please refer to this reference for the entry.}\label{table:first}
\end{table}

%% file: figdiagrama_partialktrees.tex
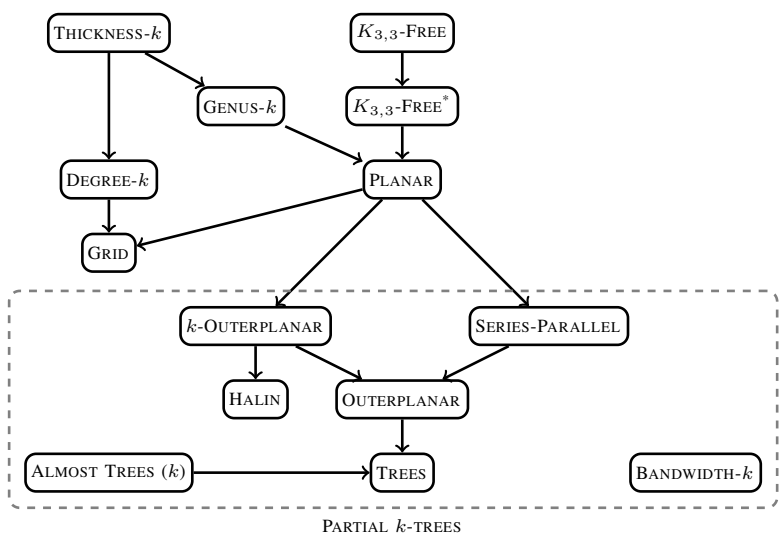
\begin{figure}[ht]\centering
  \begin{tikzpicture}[scale=0.65]
  \pgfsetlinewidth{1pt}
  \tikzset{class/.style={rounded corners,minimum size=0.5cm, draw, inner sep=2pt,font=\scriptsize}}

  \node[class] (planar) at (0,0) {\className{Planar}};
  \node[class] (k33) at (0,1.5) {\className{$K_{3,3}$-Free\textsuperscript{*}}};
  \node[class] (k33free) at (0,3) {\className{$K_{3,3}$-Free}};
  \node[class] (thick) at (-6,3) {\className{Thickness-$k$}};
  \node[class] (genus) at (-3.3,1.5) {\className{Genus-$k$}};

  \node[class] (grid) at (-6,-1.5) {\className{Grid}};
  \node[class] (degree) at (-6,0) {\className{Degree-$k$}};
  
  \node[class] (kouter) at (-3,-3) {\className{$k$-Outerplanar}};
  \node[class] (series) at (3,-3) {\className{Series-Parallel}};

  \node[class] (halin) at (-3,-4.5) {\className{Halin}};
  \node[class] (outerplanar) at (0,-4.5) {\className{Outerplanar}};
  \node[class] (tree) at (0,-6) {\className{Trees}};

  \node[class] (almosttree) at (-6,-6) {\className{Almost Trees ($k$)}};
  \node[class] (bandwidth) at (6,-6) {\className{Bandwidth-$k$}};

  \draw[->] (k33free)--(k33); 
  \draw[->] (k33)--(planar); 
  \draw[->] (genus)--(planar); 
  \draw[->] (planar)--(kouter); 
  \draw[->] (planar)--(grid); 
  \draw[->] (planar)--(series);
  \draw[->] (kouter)--(halin); 
  \draw[->] (kouter)--(outerplanar);
  \draw[->] (series)--(outerplanar);
  \draw[->] (outerplanar)--(tree);
  \draw[->] (degree)--(grid);
  \draw[->] (thick)--(degree);
  \draw[->] (thick)--(genus);
  \draw[->] (almosttree)--(tree);

\node[draw=gray,rounded corners,dashed,inner sep=2mm,label=below:{\scriptsize \className{Partial $k$-trees}},fit=(almosttree) (kouter) (bandwidth)] (partialtree) {};
\end{tikzpicture}
\caption{Containment relations for classes from [OG], where, in particular, the subclasses of \className{Partial $k$-Trees} are highlighted. A graph class \textsc{Class A} has an arrow to a graph class \textsc{Class B} if \textsc{Class A} contains \textsc{Class B}. (Adapted from [OG].)
}\label{fig:classesfirst}
\end{figure}

%% file: figdiagrama_perfect.tex
\begin{figure}[ht]\centering
\begin{tikzpicture}[scale=0.70]
  \pgfsetlinewidth{1pt}
  \tikzset{class/.style={rounded corners, minimum size=0.5cm, draw, inner sep=2pt,font=\scriptsize}}

\node[class] (perfect) at (0,0) {\className{Perfect}};
\node[class] (comp) at (-2,-2) {\className{Comparability}};
\node[class] (chordal) at (3,-2) {\className{Chordal}};
\draw[->] (perfect)--(chordal);
\draw[->] (perfect)--(comp);

\node[class] (split) at (2,-3.5) {\className{Split}};
\node[class] (strong) at (5.5,-3.5) {\className{Strongly Chordal}};
\node[class] (undir) at (3,-5) {\className{Undirected Path}};
\draw[->] (chordal)--(split);
\draw[->] (chordal)--(strong);
\draw[->] (chordal)--(undir);

\node[class] (dir) at (4,-7) {\className{Directed Path}};
\draw[->] (undir)--(dir);
\draw[->] (strong)--(dir);

\node[class] (bip) at (0,-5) {\className{Bipartite}};
\node[class] (tree) at (0,-7) {\className{Tree}};
\draw[->] (comp)--(bip);
\draw[->] (bip)--(tree);
\draw[->] (dir)--(tree);

\node[class] (perm) at (-3,-3.5) {\className{Permutation}};
\node[class] (cograph) at (-3,-5) {\className{Cograph}};
\draw[->] (comp)--(perm);
\draw[->] (perm)--(cograph);

\node[class] (circle) at (-3,-8.5) {\className{Circle}};
\node[class] (circular) at (2,-8.5) {\className{Circular Arc}};

\node[class] (PCA) at (-0.5,-10) {\className{Proper Circular Arc}};
\node[class] (interval) at (5.5,-10) {\className{Interval}};

\draw[->] (dir)--(interval);
\draw[->] (circular)--(interval);
\draw[->] (circular)--(PCA);
\draw[->] (circle)--(PCA);
\end{tikzpicture}

\caption{Containment relations for classes from [OG], our target class is \className{Undirected Path}. (Adapted from [OG].)}\label{fig:classessecond}
\end{figure}
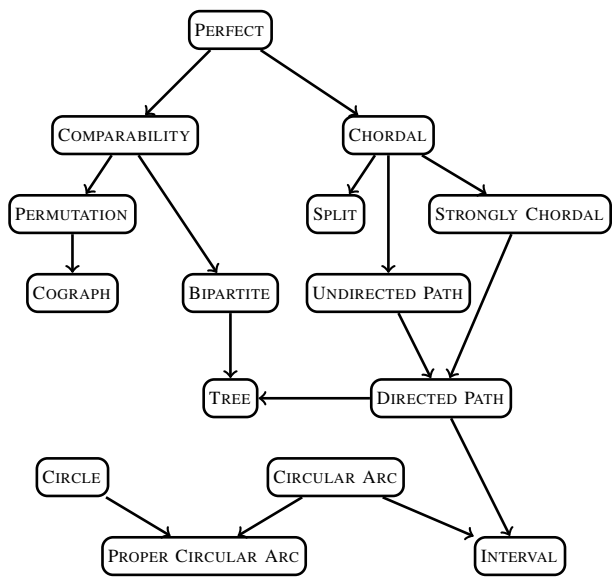

%% file: 02-Parameterized.tex
\section{The Parameterized Complexity of Hard Problems}\label{sec:parameterized}

In this section, we discuss in detail the new Table~\ref{table:second}. 
We also further discuss some of the differences that arise between our updated Table~\ref{table:first} and the table presented in [OG] thirty-five years ago. 
We start by giving some basic definitions of parameterized complexity and its complexity hierarchy classes. After that, we discuss each of the 11 columns separately.


\input{parameterized_table}

\paragraph{Parameterized Complexity}

We refer the reader to~\cite{DF.99,Niedermeier2006,DF.13,C.etal.15,Fomin2019kernelization} for all basic formal definitions, as well as many techniques employed in the parameterized complexity theory.
Formally, given a fixed finite alphabet $\Sigma$, a language $L\subseteq \Sigma^*\times \mathbb{N}$ is called a \emph{parameterized problem}; given an instance $(x,\param)\in L$, we call {\param} the \emph{parameter}. Also, we denote the \emph{size} of an instance $(x,\param)$ by $|(x,\param)|$. Observe that each possible parameter defines a different parameterized problem; for example, when considering the \problemName{Clique} problem on graphs, it can be parameterized by the size of the desired clique, or by the maximum degree of the input graph. In both cases, the input consists of a graph $G$, an integer $c$, and the corresponding parameter, and the problem consists of deciding whether $G$ has a clique of size at least $c$, except that in the former the parameter is also~$c$, while in the latter the parameter is the maximum degree $\Delta(G)$. When the parameterized problem is a decision problem having as parameter the size of the solution, we say that the problem is parameterized by the \emph{natural parameter}. Table~\ref{table:second} is filled taking into account the natural parameter, whenever possible. We give more details about this when analyzing each of the 11 columns.

We say that a parameterized problem $L$ is \emph{fixed parameter tractable} (from now on denoted by \FPT) if there exists an algorithm ${\cal A}$ that solves $L$ on input $(x,\param)$ in time $f(\param)\cdot |(x,\param)|^{\BigOh{1}}$, where $f$ is a computable function. In this case, the algorithm ${\cal A}$ is said to be an {\FPT} algorithm for $L$, and we also use {\FPT} to denote the set of {\FPT} problems. Observe that {$\P\subseteq\FPT$}.

Intuitively, one could describe the {\FPT} class as the parallel, in the parameterized complexity theory, of the {\P} class in the traditional complexity theory. Another ``approachable'' class is that of the slice-wise polynomial. We say that a parameterized problem $L$ is \emph{slice-wise polynomial} (denoted by {\XP}) if there exists an algorithm that solves $L$ in time $f(\param)\cdot |(x,\param)|^{g(\param)}$, where $f$ and $g$ are two computable functions. Observe that, for each fixed value of $\param$, this is a polynomial algorithm.

Concerning a parallel of the $\NP$-complete class, there are two main hard classes in the parameterized complexity, the \paraNP-complete and the \W-hard. 
A parameterized problem $L$ is \emph{\paraNP-complete} if it is \NP-complete for some fixed value of the parameter $\param$. For instance, the \problemName{Vertex Coloring} problem is \paraNP-complete when parameterized by the number of colors. Note that, unless ${\sf P}= \NP$, a \paraNP-complete problem cannot be \XP, hence it cannot be {\FPT} either. 

Now, before defining our last parameterized complexity class, we need another definition. Given an instance $(x,\param)$ of a parameterized problem $L$, a \emph{parameterized reduction} from $L$ to another parameterized problem $L'$ is an algorithm that computes, in time $f(\param) \cdot |x|^{\BigOh{1}}$ for some computable function $f$, an equivalent instance $(x',\param')$ of $L'$ such that $\param'\le g(\param)$ for some computable function $g$. 
The class of $\W$-hard problems can be formally defined based on a hierarchy of nested classes called $\W[i]$, for each $i\in \mathbb{N}\setminus\{0\}$. 
However, for our purposes it suffices to define the $\W[1]$-hard and $\W[2]$-hard classes in terms of their ``base'' problems; think of it as defining the $\NP$-hard class in terms of $\SAT$. A parameterized problem $L$ is \emph{\W[1]-hard} if there is a parameterized reduction from \problemName{Clique}, parameterized by the size of the clique, to $L$; and it is \emph{W[2]-hard} if there is a parameterized reduction from \problemName{Dominating Set}, parameterized by the size of the dominating set, to $L$. 
Observe that a parameterized version of an $\NP$-hard problem can be classified in any of these classes, unless of course $\P=\NP$, in which case all the classes collapse to $\P$.

Tree decompositions are an important tool in the parameterized complexity theory, as well as in the traditional computational complexity, since good algorithms can often be obtained for graphs with bounded treewidth, and also because even graphs with unbounded treewidth can sometimes be approached by applying the bidimensionality technique (see e.g.~\cite{C.etal.15}). Since the publication of [OG], in addition to treewidth, many other width parameters have been introduced (see~\cite{V.thesis} for a nice Hasse diagram containing 32 graph parameters; we also refer to the survey~\cite{G.17}). In particular, the clique-width~\cite{KLM.09} and the mim-width~\cite{V.thesis} parameters are of special interest to us, since they are bounded for some of our proposed classes, and because some of the proposed problems can be solved in polynomial time when these parameters are bounded. More specifically, \className{Partial $k$-trees} and \className{Cographs} have bounded clique-width~\cite{CO.00}, while the following have bounded mim-width: graphs with bounded clique-width~\cite{V.thesis}, \className{Permutation}~\cite{BV.13},  \className{Circular Arc}~\cite{BV.13}, \className{Directed Path}~\cite{BHMW.10} (this is because they are a subclass of leaf powers~\cite{JKST.19}). Regarding the proposed problems, the following can be solved in polynomial time on graphs with bounded mim-width, provided a branch decomposition of bounded mim-width is given: \problemName{Independent Set}, \problemName{Dominating Set}~\cite{BTV.13}, and \problemName{Steiner Tree}~\cite{BK.19}; while the following can always be solved in polynomial time on graphs with bounded clique-width, since a construction sequence of bounded clique-width can be found in polynomial (and even $\FPT$) time~\cite{OS.06,HO.08,O.08}: \problemName{Clique}~\cite{CMR.00}, \problemName{Chromatic Number}~\cite{KR.03}, \problemName{Hamiltonian Circuit}, \problemName{Maximum Cut}~\cite{W.94}, and \problemName{Clique Partition}~\cite{R.07}. 
Although the problem of finding in polynomial-time a branching decomposition of bounded mim-width is still open for general graphs with bounded mim-width, it has been proven to be polynomial-time solvable for some graph classes, including \className{Permutation} and \className{Circular Arc} graphs~\cite{BV.13}. 
These observations help to solve an issue about the complexity of \problemName{Steiner Tree} restricted to \className{Circular Arc} and \className{Proper Circular Arc}, which we discuss later on. Also, because \problemName{Maximum Cut} is $\NP$-complete on \className{Interval} graphs~\cite{ABMR.20}, and this is a subclass of \className{Circular Arc}, it follows that the problem is $\NP$-complete on graphs with bounded mim-width, which is in contrast with the complexity of the problem restricted to bounded clique-width~\cite{W.94}.


\medskip

We now discuss Tables~\ref{table:first} and~\ref{table:second}, dividing the discussion by the problems.

\paragraph{\problemName{Membership}} The entries $\P$ are inherited from Table~\ref{table:first}; hence the only class that can be further refined in the parameterized complexity is the class \className{Thickness-$k$}. It is known that deciding whether a given graph $G$ has thickness at most $k$ is $\NP$-complete even if $k=2$ (for $k=1$ it coincides with deciding planarity, which is polynomial)~[OG]. We therefore get that, considering $k$ as the parameter, the related parameterized problem is $\paraNP$-complete.

We mention that the reference cited by [OG] for \className{Partial $k$-trees} was a technical report that now has a published version~\cite{ACP.87}.

\paragraph{\problemName{Independent Set}}

The natural parameter considered is the size of the desired independent set.
The problem is trivially in $\XP$ in general: by enumerating all vertex subsets of size $\param$, one can readily check in time $\BigOh{n^{\param}}$ whether a graph on $n$ vertices admits an independent set of size at least~$\param$. 
Nevertheless, \problemName{Independent Set} is unlikely to be \FPT, since it is known to be \W[1]-hard~\cite{DF.99}. 
In fact, by considering the complement graph, we obtain that \problemName{Independent Set} and \problemName{Clique}, for general graphs, are equivalent to each other from the parameterized complexity perspective.  
On the positive side, as can be seen in Table~\ref{table:second}, the problem is {\FPT} for \className{Planar}~\cite{Niedermeier2006}, \className{Genus-$k$}~\cite{Chen2007} and \className{Degree-$k$}~\cite{DF.99} graphs. 
More generally, the problem is also {\FPT} for \className{Thickness-$k$} graphs~\cite{Khot2000}: This follows from the fact that \problemName{Independent Set} is {\FPT} when restricted to graphs with bounded clique number, as first observed in~\cite{Khot2000} as a special case of a more general framework \emph{c.f.}~\cite{Raman2008,Fomin2019kernelization}. 
Additionally, \problemName{Independent Set} is also {\FPT} for  \className{$K_{3,3}$-Free\textsuperscript{*}}.  
Indeed, it is well known that, if a graph $H$ of maximum degree at most $3$ is a minor of a graph $G$, then $H$ is a topological minor of $G$ as well \emph{c.f.}~\cite{Diestel2017,7331}. Thus, \className{$K_{3,3}$-Free\textsuperscript{*}} coincides with the class \className{$K_{3,3}$-Minor-Free\textsuperscript{*}} of graphs that do not contain $K_{3,3}$ as a minor. Therefore, since \problemName{Independent Set} was proven to be {\FPT} for \className{$K_{3,3}$-Minor-Free\textsuperscript{*}} graphs~\cite{Demaine2002,Demaine2004}, we obtain that the problem is also {\FPT} for \className{$K_{3,3}$-Free\textsuperscript{*}}. 
We remark that, although \problemName{Independent Set} is {\FPT} for \className{$K_{3,3}$-Free\textsuperscript{*}} graphs, it remains $\W[1]$-hard for the larger class \className{$K_{3,3}$-Free}. 
This latter result follows from the fact that \problemName{Independent Set} was proven to be $\W[1]$-hard for another subclass of \className{$K_{3,3}$-Free}, the  \className{$C_{4}$-free} graphs~\cite{Bonnet2019}, which consists of graphs that do not contain $C_4$ as an induced subgraph.


We observe that, for the entry \className{Partial $k$-Trees}, we cite the reference~\cite{APDAM89}, which is different from the reference~\cite{ANS.80} cited in [OG]. 
This is because, upon checking~\cite{ANS.80}, we were not able to find any mention to the \problemName{Independent Set} problem restricted to \className{Partial $k$-Trees}.


An interesting fact is that D.S.~Johnson mentions, in the caption of his summary table, that \problemName{Vertex Cover} was not included as a column because its complexity will always be the same as the complexity of \problemName{Independent Set}.
While this is true for traditional complexity theory, one could say that \problemName{Vertex Cover} is a canonical problem in $\FPT$ since it is {\FPT} in general~\cite{CKX.06} and many of the known techniques can be successfully applied to it, whereas \problemName{Independent Set} can be regarded as a canonical $\W[1]$-hard problem.


\paragraph{\problemName{Clique}} Analogously, the natural parameter considered is the size of the desired clique. All entries here are in $\P$, except for \className{Claw-Free} graphs, which is \FPT~\cite{CPPPW.11}.

\paragraph{\problemName{Partition into Cliques}} 

The problem has as input a graph $G$ and a positive integer $\param$, and it consists of deciding whether the vertex set of $G$ can be partitioned into $\param$ disjoint cliques.  
The natural parameter considered is the integer $\param$. 
One can straightforwardly verify that \problemName{Partition into Cliques} is polynomially equivalent to \problemName{Chromatic Number} by considering the complement graph $\overline{G}$ of the input graph $G$, \text{i.e.} the problem of deciding whether $\overline{G}$ can be proper colored with at most $\param$ colors. 
Nevertheless, it is worth mentioning that the respective particular cases of \problemName{Partition into Cliques} and \problemName{Chromatic Number} restricted to specific graph classes are not necessarily polynomially equivalent to each other, as usually these graph classes are not closed under taking the complement.  
As an example, observe that in Table~\ref{table:first}, \problemName{Partition into Cliques} is in $\P$ for \className{Circular Arc} graphs, while \problemName{Chromatic Number} remains $\NP$-complete for \className{Circular Arc} graphs.

The problem is trivially $\FPT$ for graphs that only contain cliques whose size can be upper-bounded by a computable function $f$ of $\param$. 
Indeed, if the size of the maximum clique of $G$ is at most $f(\param)$, then either $G$ contains at most $f(\param)\cdot \param$ vertices, or the vertex set of the input graph cannot be partitioned into at most $\param$ disjoint cliques, and thus we are dealing with a \no-instance. 
Based on that, we immediately obtain that \problemName{Partition into Cliques} is $\FPT$ for \className{Planar} and \className{$K_{3,3}$-Free\textsuperscript{*}} graphs (observe that the latter graphs cannot have cliques of size bigger than~5).
Additionally, one can verify that \className{Thickness-$k$}, \className{Genus-$k$} and \className{Degree-$k$} graphs also have cliques whose size depend only on $k$, and since $k$ is constant by definition, we get that these graphs have maximum clique bounded by a constant value. 
As a result, we obtain that \problemName{Partition into Cliques} is also {\FPT} for \className{Thickness-$k$}, \className{Genus-$k$} and \className{Degree-$k$} graphs. 

On the other hand, \className{Circle} graphs may have cliques of unbounded size. 
J.~Keil and L.~Stewart proved that \problemName{Partition into Cliques} is {\XP} for \className{Circle} graphs~\cite{keil2006}, however it remains unknown whether the problem is $\FPT$ in this class. 
The $\paraNP$-completeness of \problemName{Partition into Cliques} for \className{Claw-Free} graphs follows from the fact that \problemName{Chromatic Number} is $\NP$-complete, even when~$\param=3$, for \className{Triangle-Free} graphs (\emph{i.e.,} graphs that do not contain $K_{3}$ as an induced subgraph)~\cite{MP.96} --- since the complement graphs of \className{Triangle-Free} graphs do not have independent sets of size larger than~2, such complement graphs are \className{Claw-free}.  

In what follows, we discuss some discrepancies between our Table~\ref{table:first} and the table presented in [OG], with respect to the \problemName{Partition into Cliques} entries.  
First, [OG] cites [GJ] for the \className{Degree-$k$} entry. However it is not immediate how the results presented in [GJ] (or in the references cited by [GJ]) lead to the \NP-completeness of \problemName{Partition into Cliques} for \className{Degree-$k$}. 
In fact, [GJ] proves that \problemName{Partition into Cliques} remains $\NP$-complete for graphs that contain no cliques of size larger than~$4$; nevertheless the graph constructed in their reduction does not have bounded maximum degree. 
For this reason, we cite~\cite{CERIOLI20082270} instead, which gives an explicit $\NP$-completeness proof for \problemName{Partition into Cliques} restricted to \className{Cubic} graphs. 
As for \className{Planar} graphs, [OG] cites~\cite{BJLSPS.90}, which actually proves that the following problem is $\NP$-hard: given a planar graph $G$, and a fixed connected outerplanar graph $H$ with at least three vertices, maximizes the number of vertices of $G$ that can be covered with copies of $H$.
This does not immediately imply that \problemName{Partition into Cliques} is \NP-complete for \className{Planar} graphs since they limit the number of vertices in each clique to~$3$ (by considering $H = K_{3}$) when a planar graph could have a partition into cliques with cliques also of size~$4$. We then cite the more explicit construction given in~\cite{Kral2001}. 
Note that this also impacts the entries \className{$K_{3,3}$-Free\textsuperscript{*}}, \className{Thickness-$k$}, and \className{Genus-$k$}, as these contain \className{Planar} graphs. 
Finally, we remark that the references cited by [OG] for \className{Line} graphs and \className{Claw-Free} graphs are actually private communications. 
Therefore, we cite~\cite{MUNARO20172208} for \className{Line} graphs, and~\cite{MP.96} for \className{Claw-Free} graphs, instead of [OG].
We remark that although the \NP-completeness of \problemName{Partition into Cliques} for such graph classes directly follows from~\cite{MUNARO20172208}, we have decided to cite the oldest known reference for each result.

A problem closely related to \problemName{Partition into Cliques} is the so-called \problemName{Clique Edge-Partition}, which is defined as follows: given a graph $G$ and a positive integer $k$, decide whether the edge set of $G$ can be partitioned into at most $k$ subsets such that each subset induces a complete subgraph of $G$. 
However, it is worth mentioning that, since the line graph of a complete graph is not necessarily a complete graph, this problem is not the same as deciding whether the vertex set of the line graph $L(G)$ of a graph $G$ can be partitioned into at most $k$ disjoint cliques. 
As a matter of fact, while \problemName{Partition into Cliques} is \paraNP-complete~\cite{MP.96}, \problemName{Clique Edge-Partition} is $\FPT$ in general~\cite{MR.08}. 



\paragraph{\problemName{Chromatic Number}} The problem has as input a graph $G$ and a positive integer $\param$, with $\param$ being the considered parameter, and it consists of deciding whether the chromatic number of $G$ is at most $\param$. 
Also one of Karp's 21 \NP-complete problems~\cite{K.72}, \problemName{Chromatic Number} is $\NP$-complete for fixed values of $\param$ for very restricted graph classes. 
Since it is $\NP$-complete to decide whether a planar graph with maximum degree~$4$ is 3-colorable~\cite{GJS.76}, we get that \problemName{Chromatic Number} is $\paraNP$-complete for \className{Planar}, \className{$K_{3,3}$-Free\textsuperscript{*}}, \className{Genus-$k$}, \className{Thickness-$k$} and \className{Degree-$k$} graphs. 
Additionally, since it is $\NP$-complete to decide whether a circle graph is $4$-colorable~\cite{Walter88}, we obtain that \problemName{Chromatic Number} is \paraNP-complete for \className{Circle} graphs. 
Finally, it is also $\NP$-complete to decide whether the line graph of a 3-regular graph is 3-colorable~\cite{H.81}, implying the  $\paraNP$-completeness for \className{Line} and \className{Claw-Free} graphs. 
On the positive side, we mention that, even though the parameterized complexity theory had not yet been defined by the time of publication of~\cite{G.etal.80}, the algorithm presented in~\cite{G.etal.80} for \className{Circular Arc} is actually an $\FPT$ algorithm.
Therefore, \problemName{Chromatic Number} parameterized by $\param$, the number of colors, is {\FPT} for \className{Circular Arc} graphs. 

As for the differences between our tables and [OG]'s Table, they cite~\cite{ANS.80} for the entry \className{Partial $k$-Trees}. However this reference does not mention the \problemName{Chromatic Number} problem, and this is why we cite~\cite{APDAM89}. As for \className{Bandwidth-$k$} graphs, they cite~\cite{MS.85}, which treats only the case $k=3$. The same happens for the entry \className{$k$-Outerplanar}~\cite{B.94}. 
Also, we could not find the reference cited by [OG] for \className{Series-Parallel} graphs~\cite{S.83}, which is the same as the one cited for \className{Halin} graphs. Nevertheless, the polynomial results for all these classes follow from the fact they all have bounded treewidth; therefore we cite~\cite{APDAM89}.


\paragraph{\problemName{Chromatic Index}} This is by far the hardest problem of the table, with the largest number of open entries, remaining open even for classes considered ``easy'' as for instance \className{Cographs}, a graph class for which all problems besides \problemName{Chromatic Index} have been classified as $\P$~\cite{KR.03}. The problem has as input a graph $G$, and a positive integer $\param$, with $\param$ being the considered parameter, and one wants to decide whether the chromatic index of $G$, denoted by $\chi'(G)$, is at most~$\param$. Observe that, since Vizing's Theorem tells us that $\chi'(G)\in\{\Delta(G), \Delta(G)+1\}$, we can then consider $\param$ as being equal to the maximum degree of $G$ as otherwise the answer is trivial. As already mentioned, deciding whether $\chi'(G)= 3$ is $\NP$-complete even for \className{Cubic} graphs~\cite{H.81}, which implies that \problemName{Chromatic Index} is $\paraNP$-complete for \className{Degree-$k$} graphs, and that it is also $\paraNP$-complete for \className{Thickness-$k$} graphs. 
In addition, L. Cai and J. A. Ellis proved that deciding whether $\chi'(G)=3$ is \NP-complete even for \className{Comparability} graphs and for \className{Line} graphs~\cite{cai1991}.
Therefore, \problemName{Chromatic Index} remains $\paraNP$-complete when restricted to \className{Comparability}, \className{Perfect}, \className{Line} and \className{Claw-Free} graphs. 

The reference cited by [OG] for \className{Series-Parallel} graphs and \className{Halin} graphs is the same as the one cited for these graphs on column \problemName{Chromatic Number}~\cite{S.83}. Again, even though we could not find the reference, the results follow from the fact that these graphs have bounded treewidth~\cite{bodlaender1990}.

\paragraph{\problemName{Hamiltonian Circuit}} Here, the size of any solution is the size of the input graph; therefore, we consider the problem of deciding whether a given graph $G$ has an hamiltonian cycle parameterized by $n=|V(G)|$. It is known that \problemName{Longest Cycle} parameterized by the size {\param} of the cycle is $\FPT$~\cite{Monien1985}. If follows that \problemName{Hamiltonian Circuit} is $\FPT$ when parameterized by $n$.

The reference cited by [OG] for \className{Series-Parallel} graphs~\cite{S.83} is the same as the one cited for these graphs on \problemName{Chromatic Number} (see comments above); thus, we cite~\cite{APDAM89} instead. Also, we were not able to access the reference for \className{Split} and \className{Chordal} graphs~\cite{CS.85}, therefore we cite reference~\cite{muller1996}. 
For \className{Circle} graphs, we find a more serious discrepancy between our Table~\ref{table:first} and the table presented in [OG]. They cite~\cite{bertossi1986} as providing a polynomial algorithm for \className{Circle} graphs. However, the paper only provides a polynomial algorithm for \className{Interval} graphs, a subclass of \className{Circle} graphs. And actually, the problem has been shown to be $\NP$-complete for \className{Circle} graphs in~\cite{D.89}.

\paragraph{\problemName{Dominating Set}} Given a graph $G$ and a positive integer $\param$, the problem consists of deciding whether $G$ has a dominating set of size at most $\param$ (a set $D\subseteq V(G)$ is \emph{dominating} if, for every vertex $v \in V(G)$, $v \in D$ or $v$ has a neighbor that belongs to $D$). The natural parameter is of course $\param$, and as previously mentioned, this is the canonical problem in the class $\W[2]$-hard. Because of this, \problemName{Dominating Set} is among the most investigated problems in the parameterized complexity theory, deserving a survey of its own. Here, we make a short compilation of the results that concern the classes of interest. 
The problem is $\FPT$ for: \className{Planar}  graphs~\cite{DF.95}, and therefore also for \className{Grid} graphs; for \className{$K_{3,3}$-Free\textsuperscript{*}} graphs~\cite{PRS.09}; for \className{Genus-$k$} graphs~\cite{EFF.04}; for \className{$k$-Degenerate} graphs~\cite{AG.09}, and therefore also for \className{Degree-$k$} graphs; and for \className{Claw-Free} graphs~\cite{CPPPW.11}, and therefore also for \className{Line} graphs. 
In addition, it is $\W[2]$-hard for \className{Split} and \className{Bipartite} graphs~\cite{Raman2008}, and therefore also for \className{Chordal},  \className{Perfect}, and \className{Comparability} graphs. Finally, it is $\W[1]$-hard for \className{Circle} graphs~\cite{BGMPST.14}. 
Observe that this problem is trivially in $\XP$, since it suffices to test all the $\BigOh{n^\param}$ subsets of size $\param$. However, this is not completely refined, and the entries for \className{Thickness-$k$} and \className{Undirected Path} graphs can be regarded as the only open cases in this column. 

Regarding the differences between our Table~\ref{table:first} and the table in~[OG], they cite~\cite{G.84} for \className{Almost Trees ($k$)}, but the paper does not seem to attack this class. Something similar happens with entry~\className{Bandwidth-$k$}, where they cite~\cite{MS.85} but the paper attack domination-related problems, but not \problemName{Dominating Set} itself. 
Nevertheless, we now know that the problem is indeed polynomial for these classes since they have bounded treewidth~\cite{APDAM89}.
Also, the entry for \className{Grids} is cited as a private communication in [OG]; this is why we provide reference~\cite{Clark1990}. As for the \className{Bipartite} and \className{Comparability} entries, we were not able to find reference~\cite{D.81}, cited by [OG], this is why we also provide~\cite{MB.87}.

\paragraph{\problemName{Maximum Cut}} This is another problem that is $\FPT$ in general. Given a graph $G$, and an integer $\param$, we consider the problem of deciding whether there exists $S\subseteq E(G)$ that separates $G$ and has size at least~$\param$, parameterized by $\param$. The best known algorithm so far runs in time $\BigOh{m+n+\param\cdot 4^\param}$~\cite{MR.99}, where $m$ denotes the number of edges and $n$ denotes the number of vertices of the input graph $G$.

We were not able to find the reference cited by [OG] for \className{Thickness-$k$} graphs~\cite{B.80}, and therefore we provide the same reference given by [GJ] for \className{Degree-$k$} graphs~\cite{Y.78}.

\paragraph{\problemName{Steiner Tree}}
Given a graph $G$, a subset $X\subseteq V(G)$, called \emph{terminal} set, and a positive integer $t$, the problem consists of deciding whether there exists a subset $S\subseteq V(G)\setminus X$ such that $|S\cup X|\le t$ and $G[S\cup X]$ is connected --- and hence $G[S\cup X]$ contains a tree subgraph $\mathcal{T}$ with $X\subseteq V(\mathcal{T})$, called a \emph{Steiner tree} of $G$ for $X$. 
The vertices in $S$ are commonly called \emph{Steiner vertices}. 
\problemName{Steiner Tree} parameterized by the number of terminal vertices $|X|$ is well-known to be $\FPT$~\cite{DW.71}, with the current best algorithm running in time $\BigOh{2^{|X|}\cdot n^{\BigOh{1}}}$~\cite{BHKK.07}, where $n$ denotes the number of vertices of the input graph $G$. 
This clearly implies that \problemName{Steiner Tree} parameterized by the natural parameter, \emph{i.e.} by the size of the sought solution $|S\cup X|\leq t$, is also $\FPT$. 
Interestingly enough, the problem is \W[2]-hard when parameterized by the number of Steiner vertices $|S|$ as shown in~\cite{MRR.08}. 
We denote by $\param$ the maximum number of Steiner vertices allowed in a given instance of the problem, and then we write \problemName{$\param$-Steiner Tree} to denote \problemName{Steiner Tree} parameterized by $\param$. 
Since the other parameterized versions of the problem are already known to be $\FPT$ for general graphs, this latter is the version considered in Table~\ref{table:second}.

The \problemName{$\param$-Steiner Tree} problem is $\FPT$ for \className{Genus-$k$} graphs~\cite{PPSL.18}, and therefore for \className{Planar} and \className{Grid} graphs; and for \className{$k$-Degenerate} graphs~\cite{JLRSS.17}, and therefore for \className{Degree-$k$} graphs. 
Also, the proof given in~\cite{Raman2008} for $\W[2]$-hardness of \problemName{Dominating Set} for \className{Split} graphs actually holds for \problemName{Connected Dominating Set}. 
Moreover, in~\cite{WFP.85}, the authors give a parameterized reduction from \problemName{Connected Dominating Set} to \problemName{$\param$-Steiner Tree} that works for any subclass of \className{Chordal} graphs, without changing the input graph. 
Therefore, based on the results presented in~\cite{Raman2008}, we obtain that \problemName{$\param$-Steiner Tree} is $\W[2]$-hard for \className{Split}, \className{Chordal} and \className{Perfect} graphs. 
In the next section, we give in Proposition~\ref{prop:SteinerBipartite} a simple reduction to prove that the problem is $\W[2]$-hardness for \className{Bipartite} graphs (and therefore also for \className{Comparability} graphs). Finally, observe that a simple $\XP$ algorithm can be obtained by simply testing all  $\BigOh{n^{\param}}$ possible vertex subsets $S\subseteq V(G)\setminus X$ of size at most $\param$.

Regarding the differences between our tables and [OG]'s Table, the reference~\cite{WC.82} cited in [OG] for \className{Outerplanar} graphs could not be found, but we mention that the reference cited in [OG] for \className{Series-Parallel}~\cite{WC.83} is indeed correct, and that it can be used for \className{Outerplanar} as well. 
Also, [OG] cites a private communication with Sch\"affer for the entries \className{Circle}, \className{Line}, and \className{Claw-Free} graphs, and cites~\cite{WFP.85} for \className{Circular Arc} graphs and \className{Proper Circular Arc}. We were not able to find any mention to \className{Circular Arc} graphs in~\cite{WFP.85}. 
Also, in his book~\cite{S.book}, Spinrad writes:
\begin{quote}
``The status of \problemName{Steiner Tree} is slightly unclear; Sch\"affer sketched a proof that this is polynomially solvable (for both \className{Circle} and \className{Circular Arc} graphs), and thus it appeared as polynomial in the table of `[OG]', though no algorithm solving the problem appears in the general literature''.
\end{quote}

Nevertheless, because \className{Circular Arc} graphs have bounded mim-width and a branch decomposition with bounded mim-width of these graphs can be computed in polynomial time~\cite{BV.13}, it follows from~\cite{BK.19} that \problemName{Steiner Tree} can be solved in polynomial time for \className{Circular Arc}, and consequently also for \className{Proper Circular Arc}. 
As for entries \className{Line} and \className{Claw-Free}, they have been recently filled~\cite{BBJPPL.arxiv}, while the situation remains the same for \className{Circle} graphs. 
We add that in~\cite{K.93}, J. Keil proves that \problemName{Connected Dominating Set} is $\NP$-complete for \className{Circle} graphs. 
Thus if a proof of polynomiality of \problemName{Steiner tree} indeed exists for \className{Circle} graphs, this will be a nice example of class that separates \problemName{Connected Dominating Set} from  \problemName{Steiner tree}. 


\paragraph{\problemName{Graph Isomorphism}} This is problem Open1 from [GJ], and perhaps the most controversial problem in Graph Theory, being regarded as the only naturally defined problem with a high chance to be an $\NP$-intermediate problem, thus having deserved a classification of its own. Given two graphs $G$ and $H$, it consists of deciding whether $G$ and $H$ are isomorphic, i.e., whether there exists a bijection of the vertex sets that preserves adjacencies. A problem is $\GI$-complete if it is equivalent in complexity to general \problemName{Graph Isomorphism}.
As it happened with \problemName{Hamiltonian Circuit}, the natural parameter here is the size of the input graph. This problem is {\FPT} in general, with the best known algorithm running in time $O^*(2^{\sqrt{n}\log n})$~\cite{BL.83}. 

There are again some discrepancies between our Table~\ref{table:first} and the table presented in [OG]. In~[OG], they cite [GJ] as a reference for the entries \className{Perfect} and \className{Chordal} graphs to be $\GI$-complete; however, \className{Chordal} graphs are cited as an open case in [GJ]. Nevertheless, these classes are indeed $\GI$-complete as proven in~\cite{LB.79} for \className{Chordal} graphs; this construction was noticed to work also for \className{Split} graphs~\cite{S.78}. Something similar happens in entries \className{Bipartite} and \className{Undirected Path} graphs, with them being cited as open cases in [GJ], instead of $\GI$-complete, as cited in [OG]. Nevertheless, these are indeed $\GI$-complete as proven in~\cite{booth1979b}. This also impacts the \className{Comparability} graphs entry. 
Finally, \problemName{Graph Isomorphism} is also \GI-complete for \className{Thickness-$k$} \emph{c.f.}~\cite{isgci}. 
Since~\cite{isgci} cites an unpublished paper due to De Biasi, we provide a proof in Proposition~\ref{prop:isoThickness}, for the sake of completeness.

%% file: parameterized_table.tex
%
%


\begin{table}[ht]\centering\small
\renewcommand{\arraystretch}{1.2}
\setlength\tabcolsep{3.6pt}
\resizebox{\linewidth}{!}{
\begin{tabular}{@{}l|ll|llllllllllllllllllll@{}}
\toprule
\textsc{Graph Class} & \multicolumn{2}{l|}{\problemName{Member}} & \multicolumn{2}{l}{\problemName{IndSet}} & \multicolumn{2}{l}{\problemName{Clique}} & \multicolumn{2}{l}{\problemName{CliPar}} & \multicolumn{2}{l}{\problemName{ChrNum}} & \multicolumn{2}{l}{\problemName{ChrInd}} & 
\multicolumn{2}{l}{\problemName{HamCir}} & 
\multicolumn{2}{l}{\problemName{DomSet}} & 
\multicolumn{2}{l}{\problemName{MaxCut}} & 
\multicolumn{2}{l}{\problemName{$\param$-StTree}} & 
\multicolumn{2}{l}{\problemName{GraphIso}} \\ \midrule

\className{Partial $k$-trees} & \Poly & [OG] & \Poly & \cite{APDAM89} & \Poly & [T] & \BPoly & \textbf{\cite{blanchette2012}} & \Poly & \cite{APDAM89} & \BPoly & \textbf{\cite{bodlaender1990}}  & \Poly & \cite{APDAM89} & \Poly & \cite{APDAM89} & \BPoly & \textbf{\cite{bodlaender1994}} & \BPoly & \textbf{\cite{korach1990}} & \BPoly & \textbf{\cite{bodlaender1990}} \\

\className{Trees/Forests} & \Poly & [T] & \Poly & [GJ] & \Poly & [T] & \Poly & [GJ] & \Poly & [T] & \Poly & [GJ] & \Poly & [T] & \Poly & [GJ] & \Poly & [GJ] & \Poly & [T] & \Poly & [GJ] \\

\className{Almost Trees ($k$)} & \Poly & [OG] & \Poly & [OG] & \Poly & [T] & \BPoly & \textbf{\cite{blanchette2012}} & \BPoly & \textbf{\cite{APDAM89}} & \BPoly & \textbf{\cite{bodlaender1990}} & \BPoly & \textbf{\cite{APDAM89}} & \Poly & \cite{APDAM89} & \BPoly & \textbf{\cite{bodlaender1994}} & \BPoly & \textbf{\cite{korach1990}} & \BPoly & \textbf{\cite{bodlaender1990}} \\

\className{Bandwidth-$k$} & \Poly & [OG] & \Poly & [OG] & \Poly & [T] & \BPoly & \textbf{\cite{blanchette2012}} & \Poly & \cite{APDAM89} & \BPoly & \textbf{\cite{bodlaender1990}} & \BPoly & \textbf{\cite{APDAM89}} & \Poly & \cite{APDAM89} & \Poly & [OG] & \BPoly & \textbf{\cite{korach1990}} & \Poly & [OG] \\

\className{Series Parallel} & \Poly & [OG] & \Poly & [OG] & \Poly & [T] & \BPoly & \textbf{\cite{blanchette2012}} & \Poly & \cite{APDAM89} & \Poly & \cite{bodlaender1990} & \Poly & \cite{APDAM89} & \Poly & [OG] & \Poly & [GJ] & \Poly & [OG] & \Poly & [GJ] \\

\className{Outerplanar} & \Poly & [OG] & \Poly & [OG] & \Poly & [T] & \Poly & [OG] & \Poly & [OG] & \Poly & [OG] & \Poly & [T] & \Poly & [OG] & \Poly & [GJ] & \Poly & [OG] & \Poly & [GJ] \\

\className{Halin} & \Poly & [OG] & \Poly & [OG] & \Poly & [T] & \Poly & [OG] & \Poly & \cite{APDAM89} & \Poly & \cite{bodlaender1990} & \Poly & [T] & \Poly & [OG] & \Poly & [GJ] & \BPoly & \textbf{\cite{winter1987}} & \Poly & [GJ] \\

\className{$k$-Outerplanar} & \Poly & [OG] & \Poly & [OG] & \Poly & [T] & \Poly & [OG] & \Poly & \cite{APDAM89} & \BPoly & \textbf{\cite{bodlaender1990}} & \Poly & [OG] & \Poly & [OG] & \Poly & [GJ] & \BPoly & \textbf{\cite{korach1990}} & \Poly & [GJ] \\

\midrule

\className{Planar} & \Poly & [GJ] & \TFPT & \cite{Niedermeier2006} & \Poly & [T] & \TFPT & [T] & \PNP & \cite{GJS.76} & \Open &  & \TFPT & {\cite{Monien1985}} & \TFPT & \cite{DF.95} & \Poly & [GJ] & \TFPT & \cite{PPSL.18} & \Poly & [GJ] \\

\className{Grid} & \Poly & [OG] & \Poly & [GJ] & \Poly & [T] & \Poly & [GJ] & \Poly & [T] & \Poly & [GJ] & \TFPT & {\cite{Monien1985}} & \TFPT & \cite{DF.95} & \Poly & [T] & \TFPT & \cite{PPSL.18} & \Poly & [GJ] \\

\className{$K_{3,3}$-Free\textsuperscript{*}} & \Poly & [OG] & \TFPT & \cite{Demaine2002} & \Poly & [T] & \TFPT & [T] & \PNP & \cite{GJS.76} & \EOpen &  & \TFPT & {\cite{Monien1985}} & \TFPT & \cite{PRS.09} & \Poly & [OG] & \TXP & [T] & \BPoly & \textbf{\cite{samir2009}} \\

\className{Thickness-$k$} & \PNP & [OG] & \TFPT & \cite{Khot2000} & \Poly & [T] & \TFPT & [T] & \PNP & \cite{GJS.76} & \PNP & \cite{H.81} & \TFPT & {\cite{Monien1985}} & \TXP & [T] & \TFPT & \cite{MR.99} & \TXP & [T] & \TFPT & {\cite{BL.83}}\\

\className{Genus-$k$} & \Poly & [OG] & \TFPT & \cite{Chen2007} & \Poly & [T] & \TFPT & [T] & \PNP & \cite{GJS.76} & \EOpen &  & \TFPT & {\cite{Monien1985}} & \TFPT & \cite{EFF.04} & \TFPT & \cite{MR.99}  & \TFPT & \cite{PPSL.18} & \Poly & [OG] \\ \midrule 

\className{Degree-$k$} & \Poly & [T] & \TFPT & \cite{DF.99} & \Poly & [T] & \TFPT & [T] & \PNP & \cite{GJS.76} & \PNP & \cite{H.81} & \TFPT & {\cite{Monien1985}}
& \TFPT & \cite{AG.09} & \TFPT & \cite{MR.99} & \TFPT & \cite{JLRSS.17} & \Poly & [OG] \\ \midrule

\className{Perfect} & \BPoly & \textbf{\cite{cornuejols2013}} & \Poly & [OG] & \Poly & [OG] & \Poly & [OG] & \Poly & [OG] & \PNP & \cite{cai1991}  & \TFPT & {\cite{Monien1985}} & \WTWO & \cite{Raman2008} & \TFPT & \cite{MR.99} & \WTWO & \cite{Raman2008} & \TFPT & {\cite{BL.83}} \\

\className{Chordal} & \Poly & [OG] & \Poly & [OG] & \Poly & [OG] & \Poly & [OG] & \Poly & [OG] & \EOpen &  & \TFPT & {\cite{Monien1985}} & \WTWO & \cite{Raman2008} & \TFPT & \cite{MR.99} & \WTWO & \cite{Raman2008} & \TFPT & {\cite{BL.83}} \\

\className{Split} & \Poly & [OG] & \Poly & [OG] & \Poly & [OG] & \Poly & [OG] & \Poly & [OG] & \EOpen &  & \TFPT & {\cite{Monien1985}} & \WTWO & \cite{Raman2008} & \TFPT & \cite{MR.99} & \WTWO & \cite{Raman2008} & \TFPT & {\cite{BL.83}} \\

\className{Strongly Chordal} & \Poly & [OG] & \Poly & [OG] & \Poly & [OG] & \Poly & [OG] & \Poly & [OG] & \EOpen &  & \TFPT & {\cite{Monien1985}} & \Poly & [OG] & \TFPT & \cite{MR.99} & \Poly & [OG] & \TFPT & {\cite{BL.83}}  \\

\className{Comparability} & \Poly & [OG] & \Poly & [OG] & \Poly & [OG] & \Poly & [OG] & \Poly & [OG] & \PNP & \cite{cai1991} & \TFPT & {\cite{Monien1985}} & \WTWO & \cite{Raman2008} & \TFPT & \cite{MR.99} & \WTWO & Prop.~\ref{prop:SteinerBipartite} & \TFPT & {\cite{BL.83}} \\

\className{Bipartite} & \Poly & [T] & \Poly & [GJ] & \Poly & [T] & \Poly & [GJ] & \Poly & [T] & \Poly & [GJ] & \TFPT & {\cite{Monien1985}} & \WTWO & \cite{Raman2008} & \Poly & [T] & \WTWO & Prop.~\ref{prop:SteinerBipartite} & \TFPT & {\cite{BL.83}} \\

\className{Permutation} & \Poly & [OG] & \Poly & [OG] & \Poly & [OG] & \Poly & [OG] & \Poly & [OG] & \EOpen &  & \BPoly & \textbf{\cite{deogun1994}}  & \Poly & [OG] & \TFPT & \cite{MR.99} & \Poly & [OG] & \Poly & [OG] \\

\className{Cographs} & \Poly & [T] & \Poly & [OG] & \Poly & [OG] & \Poly & [OG] & \Poly & [OG] & \EOpen &  & \Poly & [OG] & \Poly & [OG] & \BPoly & \textbf{\cite{bodlaender1994}} & \Poly & [OG] & \Poly & [OG] \\ \midrule

\className{Undirected Path} & \Poly & [OG] & \Poly & [OG] & \Poly & [OG] & \Poly & [OG] & \Poly & [OG] & \EOpen &  & \TFPT & {\cite{Monien1985}} & \TXP & [T] & \TFPT & \cite{MR.99} & \TXP & [T] & \TFPT & {\cite{BL.83}} \\

\className{Directed Path} & \Poly & [OG] & \Poly & [OG] & \Poly & [OG] & \Poly & [OG] & \Poly & [OG] & \EOpen &  & \TFPT & {\cite{Monien1985}}  & \Poly & [OG] & \TFPT & \cite{MR.99}  & \Poly & [OG] & \BPoly & \textbf{\cite{Babel1996}}  \\

\className{Interval} & \Poly & [OG] & \Poly & [OG] & \Poly & [OG] & \Poly & [OG] & \Poly & [OG] & \EOpen &  & \Poly & [OG] & \Poly & [OG] & \TFPT & \cite{MR.99}  & \Poly & [OG] & \Poly & [OG] \\

\className{Circular Arc} & \Poly & [OG] & \Poly & [OG] & \Poly & [OG] & \Poly & [OG] & \TFPT & {\cite{G.etal.80}} & \EOpen &  & \BPoly & \textbf{\cite{shih1992n}} & \Poly & [OG] & \TFPT & \cite{MR.99}  & \Poly & \cite{BK.19} & \BPoly & \textbf{\cite{Krawczyk2019}} \\

\className{Circle} & \Poly & [OG] & \Poly & [GJ] & \Poly & [OG] & \TXP & {\cite{keil2006}} & \PNP & \cite{Walter88} & \EOpen &  & \TFPT & {\cite{Monien1985}} & \WONE & {\cite{BGMPST.14}}  & \TFPT & \cite{MR.99} & \Poly & [\pc{OG}] & \BPoly & \textbf{\cite{Kalisz2019}} \\

\className{Proper Circ. Arc} & \Poly & [OG] & \Poly & [OG] & \Poly & [OG] & \Poly & [OG] & \Poly & [OG] & \EOpen &  & \Poly & [OG] & \Poly & [OG] & \TFPT & \cite{MR.99}  & \Poly & \cite{BK.19} & \BPoly & \textbf{\cite{LSS.08}}  \\

\className{Edge (or Line)} & \Poly & [OG] & \Poly & [GJ] & \Poly & [T] & \ONPh & \cite{MUNARO20172208} & \PNP & \cite{H.81} & \PNP & \cite{cai1991}  & \TFPT & {\cite{Monien1985}} & \TFPT & \cite{CPPPW.11} & \BPoly & \textbf{\cite{guruswami1999}}  & \TXP & [T] & \TFPT & {\cite{BL.83}}\\

\className{Claw-Free} & \Poly & [T] & \Poly & [OG] & \TFPT & \cite{CPPPW.11} & \PNP & \cite{MP.96} & \PNP & \cite{H.81} & \PNP & \cite{cai1991}  & \TFPT & {\cite{Monien1985}} & \TFPT & \cite{CPPPW.11} & \TFPT & \cite{MR.99} & \TXP & [T] & \TFPT & \cite{BL.83} \\ \bottomrule
\end{tabular}
}

\caption{The parameterized \NP-Completeness Column: An Ongoing Guide table revised for the 21st century. The parametrized puzzle is to classify every \COpen\xspace  entry, every \COpen? entry and every {\CNPC} entry into {\FPT} = Fixed parameter tractable, \CWONE = \W[1]-hard, \CWTWO = \W[2]-hard, and {\CparaNP} = \paraNP-complete, where the considered parameterization is with respect to the natural parameter of each corresponding problem. 
We highlight as {\COpen*} the {\CNPC} entry of Table~\ref{table:first} that constitutes the parameterized puzzle, for which so far we were not able to provide a parameterized complexity classification.}\label{table:second}
\end{table}

%% file: 03-Reductions.tex
\section{Some Simple Reductions}
\label{sec:simplered}


For the sake of completeness, here we present two simple proofs. First, we prove that 
\problemName{$\kappa$-Steiner Tree} is \W[2]-hard when restricted to \className{Bipartite} graphs. 
Indeed, this result follows from a standard parameterized reduction from \problemName{Dominating Set}, described in Proposition~\ref{prop:SteinerBipartite}. 
We remark that Raman and Saurabh  present a similar reduction to prove that \problemName{Dominating Set} is \W[2]-hard for \className{Bipartite} graphs~\cite{Raman2008}. 

\begin{proposition}\label{prop:SteinerBipartite}
    \problemName{$\param$-Steiner Tree} remains $\W[2]$-hard for \className{Bipartite} graphs. 
\end{proposition}
\begin{proof}
Let $I=(G,\param)$ be an instance of \problemName{Dominating Set}.
We let $I' = (G', X, \param)$ be the instance of \problemName{$\param$-Steiner Tree} such that $G'$ is defined as follows: 
\begin{itemize}
    \item $V(G') = \set{r} \cup \set{v' \setst v \in V(G)} \cup V(G)$, where $r$ denotes a new vertex, and we add a new vertex $v'$ for each $v\in V(G)$;
    \item $E(G') = \set{rv \setst v \in V(G)} \cup \set{v'u \setst u \in N_{G}[v],\, u,v \in V(G)}$; and 
\end{itemize}
the terminal set is defined as $X=\set{r}\cup\set{v' \setst v \in V(G)}$. 
Note that $X$ and $V(G)$ are independent sets of $G'$. 
Thus, $G'$ is a bipartite graph. 

Suppose that $G$ admits a dominating set $D \subseteq V(G)$ of size at most $\param$. 
It is not hard to check that $D \cup X$ induces a connected subgraph of $G'$.
Therefore, $I'$ is a \yes-instance of \problemName{$\param$-Steiner Tree}. 

Conversely, suppose that $G'$ admits a Steiner tree $T$ for $X$ such that $\abs{V(T)\setminus X} \leq \param$. 
Note that neighbors of $X$ in $G'$ belong to $V(G)$.
Thus, $V(T) \cap V(G)$ is a dominating set of $G$, otherwise either $T$ would not be connected, or there would exist some terminal vertex belonging to $X\setminus\set{r}$ that is not in $T$. 
Moreover, since $\abs{V(T)\setminus X} \leq \param$, we obtain that $\abs{V(T) \cap V(G)} \leq \param$. 
Therefore, $V(T)\cap V(G)$ is a dominating set of $G$ of size at most $\param$, and $I$ is a \yes-instance of \problemName{Dominating Set}. 
\end{proof}

Now, we present a proof that \problemName{Graph Isomorphism} is $\GI$-complete when restricted to \className{Thickness-$k$} graphs. This result actually follows from a simple adaptation of an argument described in~\cite{24882} by De Biasi, which we present in  Proposition~\ref{prop:isoThickness}. 
The \emph{subdivision} of a graph $G$ is defined as the graph $\mathfrak{s}(G)$ obtained from $G$ by replacing each edge $e = uv\in E(G)$ with the path $\sequence{u,w_{e},v}$, where $w_{e}$ denotes a new vertex. 
More formally, $\mathfrak{s}(G)$ is the graph with vertex set $V(\mathfrak{s}(G))=V(G)\cup\set{w_{e}\setst e \in E(G)}$ and edge set $E(\mathfrak{s}(G)) = \set{uw_{e},w_{e}v \setst e=uv \in E(G)}$. 

\begin{lemma}\label{lemma:subdivisionThickness}
	For each graph $G$, the subdivision $\mathfrak{s}(G)$ of $G$ has thickness at most~$2$.
\end{lemma}
\begin{proof}
    Assume without loss of generality that $V(G)=\set{v_{1},\ldots,v_{n}}$, for some positive integer~$n$. 
	Let $H_1$ and $H_2$ be the  spanning subgraphs of $\mathfrak{s}(G)$ defined as follows: for each edge $e = v_{i}v_{j}\in E(G)$ with $i < j$, add the edge $v_{i}w_{e}$ to $H_{1}$ and add the edge $w_{e}v_{j}$ to $H_{2}$. 
	Note that, for each $e \in E(G)$, the degree of $w_{e}$ in $H_{1}$, and in $H_{2}$, is exactly $1$.
	Moreover, $V(G)$ is an independent set in $H_{1}$, and in $H_{2}$. 
	Thus, $H_{1}$ and $H_{2}$ are forests whose components are stars, which implies that $H_{1}$ and $H_{2}$ are planar graphs. 
	Therefore, since $\mathfrak{s}(G)=H_{1}\cup H_{2}$, we obtain that $\mathfrak{s}(G)$ has thickness at most~$2$. 
\end{proof}

\begin{proposition}\label{prop:isoThickness}
	\problemName{Graph Isomorphism} is \GI-complete for \problemName{Thickness-$k$}.
\end{proposition}
\begin{proof}
	Let $G_{1}$ and $G_{2}$ be two arbitrary graphs.
	It follows from Lemma~\ref{lemma:subdivisionThickness} that  $\mathfrak{s}(G_{1})$ and $\mathfrak{s}(G_{2})$ have thickness at most~$2$. 
	Moreover, one can easily verify that $G_{1}$ and $G_{2}$ are isomorphic if and only if $\mathfrak{s}(G_{1})$ and $\mathfrak{s}(G_{2})$ are isomorphic.
\end{proof}

%% file: 04-Steiner.tex
\section{Steiner Tree for Undirected Path Graphs}\label{sec:ST}
In this section, we prove that the \problemName{Steiner Tree} problem is \NP-complete for \className{Undirected Path} graphs, 
which provides a full dichotomy Polynomial versus \NP-complete for the \problemName{Steiner Tree} column.
 Our proof holds even if the input graph has diameter~3, and we show that \problemName{Steiner Tree} is in $\P$ when restricted to \className{Undirected Path} graphs of diameter~2, thus getting another dichotomy for the problem in terms of the diameter.

We mention that Spinrad writes in his book~\cite{S.book} that he was unable to find any work on the \problemName{Steiner Tree} problem restricted to \className{Undirected Path} graphs, but that Dieter Kratsch told him this should be \NP-complete as a simple extension of a proof that \problemName{Connected Dominating Set} is \NP-complete for \className{Undirected Path} graphs. 
Haynes et al. cite a paper~\cite{KLM.97}, submitted in 1997, in their book~\cite{HHS.book},
where the \NP-completeness proof of \problemName{Connected Dominating Set} supposedly appears. However, we were not able to find any version of~\cite{KLM.97}. 
Thus, in order to fill this gap, we provide a non-trivial adaptation of the \NP-completeness proof presented in~\cite{BJ.82} for the \problemName{Dominating Set} problem restricted to \className{Undirected Path} graphs, which finally explicitly shows that the \problemName{Connected Dominating Set} problem restricted to \className{Undirected Path} graphs is indeed \NP-complete.  
Then, we use a transformation by White et al.~\cite{WFP.85} between the \problemName{Steiner Tree} and the \problemName{Connected Dominating Set} problems to obtain the desired result as a corollary. 

A closely related variant of \problemName{Connected Dominating Set} that should be mentioned is the so-called \problemName{Total Dominating Set} problem, which, rather than a dominating set inducing a connected subgraph, simply requires a dominating set having no isolated vertices. 
Through a different non-trivial adaptation of the proof presented in~\cite{BJ.82}, \problemName{Total Dominating Set} restricted to \className{Undirected Path} graphs was proven to be \NP-complete~\cite{Laskar1984}. 
However, it is worth noticing that the construction described in~\cite{Laskar1984} cannot be used so as to further obtain the \NP-completeness of \problemName{Connected Dominating Set} for \className{Undirected Path} graphs.
Therefore, we emphasize the merit of our contribution. 


\medskip

We start by giving some formal definitions. A \emph{chordal} graph can also be described as the intersection graph of subtrees of a tree: given a tree $T$, each vertex $u$ of $G$ corresponds to a subtree $T_u$ of $T$, and $uv\in E(G)$ if and only if $V(T_u)\cap V(T_v)\neq \emptyset$. We call $(T,\{T_u\}_{u\in V(G)})$ a \emph{tree model of $G$}. One can verify that a tree decomposition of $G$ of width $\omega(G)$ can be obtained from this tree model. 
%
The subclasses of \className{Undirected Path}, \className{Directed Path} and \className{Interval} graphs can be derived from this definition as follows. An \emph{undirected path graph} is a chordal graph that has a tree model $(T,\{T_u\}_{u\in V(G)})$ where each $u\in V(G)$ corresponds to a subpath of $T$. A \emph{directed path graph} is an undirected path graph that has a tree model $(T,\{T_u\}_{u\in V(G)})$ such that $T$ is rooted at a vertex $r$, and every subpath $T_u$ is from a node $t\in V(T)$ to a node $t'\in V(T)$ where $t$ belongs to the $(r, t')$-path
of $T$. Finally, an \emph{interval graph} is a chordal graph that has a tree model $(T,\{T_u\}_{u\in V(G)})$ where $T$ is a path. From Figure~\ref{fig:classessecond}, we know that these classes are nested. 

\medskip

We recall that, given a graph $G$, a subset $D\subseteq V(G)$ is a \emph{dominating set} of $G$ if, for every vertex $v\in V(G)\setminus D$, $N_{G}(v)\cap D\neq \emptyset$. Additionally, $D$ is said to be \emph{connected} if $G[D]$ is a connected subgraph of $G$. 
Given also a subset $X\subseteq G$ of \emph{terminals}, a \emph{Steiner tree} of $G$ for $X$ is a tree subgraph ${\cal T}$ of $G$ such that 
$X\subseteq V({\cal T})$. 
Next, we formally state the \problemName{Steiner Tree} and \problemName{Connected Dominating Set} problems. Although the usual question for \problemName{Steiner Tree} asks for the minimum tree, it is more convenient for our reduction to ask for the minimum set of non terminal vertices. 
Notice that this gives a polynomially equivalent problem. 


\begin{myproblem}
  \problemtitle{\problemName{Steiner Tree}}
  \probleminput{A connected graph $\graph$, a non-empty subset $X\subseteq \gvset{\graph}$, and a positive integer $k$.}
  \problemquestion{Does there exist a subset $S\subseteq \gvset{\graph}\setminus X$ with $\abs{S} \leq k$, such that $\graph[S\cup X]$ is connected?}
\end{myproblem}

\begin{myproblem}
  \problemtitle{\problemName{Connected Dominating Set}}
  \probleminput{A graph $\graph$ and a positive integer $k$.}
  \problemquestion{Does there exist a subset $\domset\subseteq \gvset{\graph}$ with $\abs{\domset} \leq k$, such that $\cnh{\graph}{\domset}=\gvset{\graph}$ and $G[D]$ is connected?}
\end{myproblem}

As we said before, we first prove that \problemName{Connected Dominating Set} is \NP-complete for \className{Undirected Path} graphs. We do this with a reduction from the following problem, which is one of Karp's 21 \NP-complete problems~\cite{K.72}.

\begin{myproblem}
  \problemtitle{\problemName{3D-Matching}}
  \probleminput{Disjoint sets $\pset$, $\qset$ and $\rset$ each of cardinality $\nvertices$, for some positive integer $\nvertices$, and a subset $\sset\subseteq \pset\times\qset\times\rset$.}
  \problemquestion{Does there exist a subset $\domset\subseteq \sset$ such that $\abs{\domset}=\nvertices$ and $\striple\cap\striple'=\emptyset$ for every two triples  $\striple,\striple'\in\domset$?}
\end{myproblem}

\begin{theorem}\label{thm:SteinerUP}
	\problemName{Connected Dominating Set} remains $\NP$-complete when restricted to undirected path graphs of diameter at most~3.
\end{theorem}
\begin{proof}
	Let $\pset=\set{\pvertex_{1},\ldots,\pvertex_{\nvertices}}$, $\qset=\set{\qvertex_{1},\ldots,\qvertex_{\nvertices}}$ and $\rset=\set{\rvertex_{1},\ldots,\rvertex_{\nvertices}}$ be disjoint sets each of cardinality $\nvertices$, for some positive integer~$\nvertices$, and let $\sset=\set{\striple_{1},\ldots,\striple_{\nedges}}$ be a subset of  $\pset\times\qset\times\rset$ of cardinality $\nedges$, for some positive integer $\nedges$. 
	We let $\instance=\tuple{\pset,\qset,\rset,\sset}$ be the instance of \problemName{3D-Matching} constituted by $\pset$, $\qset$, $\rset$ and $\sset$. 
	Then, we let $\graph$ be the graph obtained from $\instance$ as follows (Figure~\ref{graphGConnDomSet} shows a tree model of the constructed graph):
	\begin{itemize}
	\item For each $\striple_{j}\in\sset$, we let $\vset_{j} = \set{\avertex_{j}, \bvertex_{j}, \cvertex_{j}, \xvertex_{j}, \yvertex_{j}, \zvertex_{j}^{1}, \zvertex_{j}^{2}, \zvertex_{j}^{3}}$. 
	We remark that, for each two sets $\striple_{j}, \striple_{\ell} \in \sset$, $\vset_{j}\cap\vset_{\ell}=\emptyset$ if and only if $j\neq\ell$; 
	\item $V(G) = \cup_{j=1}^m \vset_j \cup \pset \cup \qset \cup \rset$;
		\item $\kclique=\bigcup_{\striple_{j}\in\sset}\set{\avertex_{j},\bvertex_{j},\cvertex_{j},\xvertex_{j}}$ is a clique of $\graph$;
		\item for each $\striple_{j}\in\sset$, $\set{\avertex_{j},\bvertex_{j},\xvertex_{j},\yvertex_{j}}$, $\set{\avertex_{j},\yvertex_{j},\zvertex_{j}^{1}}$, $\set{\bvertex_{j},\yvertex_{j},\zvertex_{j}^{2}}$ and $\set{\cvertex_{j},\xvertex_{j},\zvertex_{j}^{3}}$ are cliques of $\graph$;
		\item for each $\pvertex_{i}\in\pset$, $\set{\pvertex_{i}} \cup \set{\avertex_{j} \setst \pvertex_{i} \in \striple_{j}, \striple_{j} \in \sset}$ is a clique of $\graph$;
		\item for each $\qvertex_{i}\in\qset$, $\set{\qvertex_{i}} \cup \set{\bvertex_{j} \setst \qvertex_{i} \in \striple_{j}, \striple_{j} \in \sset}$ is a clique of $\graph$; 
		\item for each $\rvertex_{i}\in\rset$, $\set{\rvertex_{i}} \cup \set{\cvertex_{j} \setst \rvertex_{i} \in \striple_{j}, \striple_{j} \in \sset}$ is a clique of $\graph$.
	\end{itemize}


    Figure~\ref{graphGConnDomSet} illustrates a tree model $\tuple{\cliquetree,\set{\tree_{\uvertex}}_{\uvertex\in\gvset{\graph}}}$ associated with the graph $\graph$. We depict inside a node $t\in V(T)$, the set of vertices of $G$ that contains node $t$ in its corresponding subtree; more formally, denoting by $X_t$ the subset of $V(G)$ drawn inside node $t$, and given $v\in V(G)$, we define $T_v$ as the subtree of $T$ induced by $\{t\in V(T)\colon v\in X_t\}$.  
	Based on $\tuple{\cliquetree,\set{\tree_{\uvertex}}_{\uvertex\in\gvset{\graph}}}$, one can verify that $\graph$ is an undirected path graph. 
	Indeed, for each vertex $\uvertex\in\gvset{\graph}$, we get that $\tree_{\uvertex}$ is a path of $\cliquetree$. 
	Moreover, one can readily verify that $\kclique$ is a dominating clique  of $\graph$. 
	Therefore, $\graph$ has diameter at most $3$. 

	\input{graphGConnDomSet}
	
	

	We now prove that $\instance$ is a \yes-instance of \problemName{3D-Matching} if and only if $\graph$ admits a connected dominating set $\domset$ of size at most $2\nedges+\nvertices$. 

	First, suppose that $\instance$ is a \yes-instance of \problemName{3D-Matching}, and let $\matching$ be a 3d-matching of $\instance$. 
	Then, we define $\domset=\set{\avertex_{j},\bvertex_{j},\cvertex_{j}\setst \striple_{j}\in\matching}\cup\set{\xvertex_{j},\yvertex_{j}\setst \striple_{j}\not\in\matching}$.
	Note that $\abs{\domset} = 3\nvertices+2(\nedges-\nvertices)=2\nedges+\nvertices$. 
	We claim that $\domset$ is a connected dominating set of $\graph$. 
	Indeed, since $\matching$ is a 3d-matching of $\instance$, we have that the following holds: 
	\begin{itemize}
		\item for each $\pvertex_{i}\in\pset$, there exists (exactly) one triple $\striple_{j}\in\matching$ such that $\pvertex_{i}\in\striple_{j}$, which implies that $\avertex_{j}\in\domset$ and that $\pvertex_{i}$ is dominated in $\graph$ by $\avertex_{j}$;

		\item for each $\qvertex_{i}\in\qset$, there exists (exactly) one triple $\striple_{j}\in\matching$ such that $\qvertex_{i}\in\striple_{j}$, which implies that $\bvertex_{j}\in\domset$ and that $\qvertex_{i}$ is dominated in $\graph$ by $\bvertex_{j}$;

		\item for each $\rvertex_{i}\in\rset$, there exists (exactly) one triple $\striple_{j}\in\matching$ such that $\rvertex_{i}\in\striple_{j}$, which implies that $\cvertex_{j}\in\domset$ and that $\rvertex_{i}$ is dominated in $\graph$ by $\cvertex_{j}$.
	\end{itemize}
	Additionally, it directly follows from the construction of $\graph$ that $\set{\avertex_{j},\bvertex_{j},\cvertex_{j}}$  dominates all vertices belonging to $\vset_{j}$ for each $\striple_{j}\in\matching$, and that $\set{\xvertex_{\ell},\yvertex_{\ell}}$ dominates all vertices belonging to $\vset_{\ell}$ for each $\striple_{\ell}\in\sset\setminus\matching$. 
	Consequently, $\domset$ is a dominating set of $\graph$.  
	To verify that $\domset$ induces a connected subgraph of $\graph$, first note that $\kclique'=\bigcup_{\striple_{j}\in\matching}\set{\avertex_{j},\bvertex_{j},\cvertex_{j}}$ induces a connected subgraph of $\graph$ since $\kclique$ is a clique of $\graph$ and $\kclique'\subseteq\kclique$. 
	In addition, it follows from the facts that $\bigcup_{\striple_{j}\in\sset}\set{\xvertex_{j}}\subseteq \kclique$, and $ \xvertex_{j}\yvertex_{j}\in\geset{\graph}$ for each $\striple_{j}\in\sset$, that the set $\kclique''=\bigcup_{\striple_{j}\in\sset\setminus\matching}\set{\xvertex_{j},\yvertex_{j}}$ induces a connected subgraph of $\graph$. 
Therefore, the result follows since $D = K'\cup K''$ is a connected subgraph of $G$.

	Conversely, suppose now that $\graph$ admits a connected dominating set $\domset$ of size at most $2\nedges+\nvertices$. 
	By the construction of $\graph$, for each $\striple_{j}\in\sset$, the following holds (see Figure~\ref{fig:Vj}):
	\begin{itemize} 
		\item $\domset \cap \set{\avertex_{j},\yvertex_{j},\zvertex_{j}^{1}} \neq \emptyset$, otherwise $\zvertex_{j}^{1}$ would not be dominated in $\graph$ by $\domset$; 
		\item $\domset \cap \set{\bvertex_{j},\yvertex_{j},\zvertex_{j}^{2}} \neq \emptyset$, otherwise $\zvertex_{j}^{2}$ would not be dominated in $\graph$ by $\domset$; 		 
		\item $\domset \cap \set{\cvertex_{j},\xvertex_{j},\zvertex_{j}^{3}} \neq \emptyset$, otherwise $\zvertex_{j}^{3}$ would not be dominated in $\graph$ by $\domset$. 
	\end{itemize}
	
    We recall that $\vset_{j} = \set{\avertex_{j}, \bvertex_{j}, \cvertex_{j}, \xvertex_{j}, \yvertex_{j}, \zvertex_{j}^{1}, \zvertex_{j}^{2}, \zvertex_{j}^{3}}$ for each $\striple_{j}\in\sset$. 
    \input{figVj}
	Then, based on the above, one can verify that $\abs{\domset \cap \vset_{j}} \geq 2$. 
	Moreover, we prove that if $\abs{\domset \cap \vset_{j}} = 2$ for some $\striple_{j}\in\sset$, then $\domset \cap \vset_{j} = \set{\xvertex_{j},\yvertex_{j}}$. 
	This follows from the fact that the only other possibilities for $\domset \cap \vset_{j}$ with $\abs{\domset \cap \vset_{j}} = 2$ would be $\domset \cap \vset_{j}=\set{\cvertex_{j},\yvertex_{j}}$ and $\domset \cap \vset_{j}=\set{\zvertex_{j}^{3},\yvertex_{j}}$. 
	However, note that, $\set{\avertex_{j},\bvertex_{j},\xvertex_{j}}$ is a separator of $\cvertex_{j}$ and $\yvertex_{j}$ in $\graph$ for each $\striple_{j}\in\sset$.  
	Thus, if $\domset \cap \vset_{j}=\set{\cvertex_{j},\yvertex_{j}}$, then $\domset$ would not induce a connected subgraph of $\graph$. 
	Analogously, $\set{\cvertex_{j},\xvertex_{j}}$ is a separator of $\zvertex_{j}^{3}$ and $\yvertex_{j}$ in $\graph$ for each $\striple_{j}\in\sset$.  
	Thus, if $\domset \cap \vset_{j}=\set{\zvertex_{j}^{3},\yvertex_{j}}$, then $\domset$ would not induce a connected subgraph of $\graph$. 
	As a result, we have that $\domset \cap \vset_{j} = \set{\xvertex_{j},\yvertex_{j}}$ whenever $\abs{\domset \cap \vset_{j}} = 2$. 
	
	We prove now that we can assume that $\domset\cap\vset_{j}=\set{\avertex_{j},\bvertex_{j},\cvertex_{j}}$ for each $\striple_{j}\in\sset$ with $\abs{\domset\cap\vset_{j}}\geq 3$. 
	Indeed, for each $\striple_{j}\in\sset$, $\set{\avertex_{j},\bvertex_{j},\cvertex_{j}}$ dominates at least the same vertices of $\graph$ as any other subset of $\vset_{j}$, which implies $\cnh{\graph}{\domset\cap\vset_{j}}\subseteq\cnh{\graph}{\set{\avertex_{j},\bvertex_{j},\cvertex_{j}}}$. 
	Moreover, since by hypothesis $\domset$ induces a connected subgraph of $\graph$, we obtain that $(\domset\setminus\vset_{j})\cup\set{\avertex_{j},\bvertex_{j},\cvertex_{j}}$ also induces a connected subgraph of $\graph$ for each $\striple_{j}\in\sset$. 
	Thus, assume without loss of generality that $\domset\cap\vset_{j}=\set{\avertex_{j},\bvertex_{j},\cvertex_{j}}$ whenever $\abs{\domset\cap\vset_{j}}\geq 3$. 

	Now, let $\matching=\set{\striple_{j} \in \sset \setst \abs{\domset\cap\vset_{j}}=3}$. 
	Based on the definition of $\matching$, we obtain that $\abs{\domset} \geq 3\abs{\matching}+2(\nedges-\abs{\matching})=2\nedges+\abs{\matching}$. 
	On the other hand, we have by hypothesis that $\abs{\domset}\leq 2\nedges+\nvertices$. 
	Consequently, $\abs{\matching}\leq \nvertices$. 
	Towards a contradiction, suppose that $\abs{\matching}<\nvertices$. 
	Let $\domset'=\bigcup_{\striple_{j}\in\sset}(\domset\cap\vset_{j})$. 
	Note that, for each $\uset\in\set{\pset,\qset,\rset}$, there are at most $\abs{\matching}$ vertices from $\uset$ that are dominated in $\graph$ by some vertex in $\domset'$. 
	As a result, there exist at least $3(\nvertices-\abs{\matching})$ distinct vertices from $\pset\cup\qset\cup\rset$ that are not dominated in $\graph$ by any vertex belonging to $\domset'$, \emph{i.e.} $\abs{(\pset\cup\qset\cup\rset)\setminus\cnh{\graph}{\domset'}}\geq 3(\nvertices-\abs{\matching})$. 
	This implies that $\domset$ must further contain each of such vertices from $\pset\cup\qset\cup\rset$ that are not dominated in $\graph$ by $\domset'$. 
	Then, we obtain that actually $$\abs{\domset} \geq 3\abs{\matching}+2(\nedges-\abs{\matching}) + 3(\nvertices-\abs{\matching})=2\nedges+\nvertices+2(\nvertices-\abs{\matching})>2\nedges+\nvertices\text{,}$$ which contradicts the hypothesis of $\domset$ being a connected dominating set of $\graph$ of size at most $2\nedges+\nvertices$. 
	Consequently, $\abs{\matching}=\nvertices$ and $\abs{\domset}=2\nedges+\nvertices$.  
	This implies that, if $\uvertex$ is a vertex in $\domset$, then $\uvertex\in\vset_{j}$ for some $\striple_{j}\in\sset$. 
	Hence, we obtain that $\matching$ is a 3d-matching of $\instance$, otherwise there would exist a vertex $\vvertex\in\pset\cup\qset\cup\rset$ such that, for every triple $\striple_{j}\in\sset$ with $\vvertex\in\striple_{j}$, $\abs{\domset\cap\vset_{j}}=2$, which would imply that $\vvertex$ is not dominated in $\graph$ by any vertex belonging to $\domset$. 
	Therefore, $\instance$ is a \yes-instance of \problemName{3D-Matching}. 
\end{proof}

\medskip
As previously mentioned, in~\cite{WFP.85}, the authors give a reduction from \textsc{Connected Dominating Set} to \textsc{Steiner Tree} that works in any subclass of \className{Chordal} graphs without changing the input graph. We then get the following corollary. 

\begin{corollary}
\problemName{Steiner Tree} is $\NP$-complete when restricted to undirected path graphs of diameter at most~3.
\end{corollary}

In~\cite{K.90}, the author proves that deciding whether an undirected path graph has a dominating clique of size at most $k$ can be done in polynomial time. We apply his result to get a dichotomy for \problemName{Steiner Tree} restricted to \className{Undirected Path} graphs in terms of the diameter of the input graph. 
For this, we first need some definitions and tool lemmas, which are presented below.

Let $G$ be a connected graph and $X\subseteq V(G)$ be a non-empty set.
We denote by $ST(G,X)$ the minimum size of a subset $S \subseteq V(G)\setminus X$ such that $S\cup X$ induces a connected subgraph of $G$. 
Throughout the remainder of this section, we assume without loss of generality that $|X|\ge 3$ and that $X$ does not induce a connected subgraph of $G$, as otherwise $ST(G, X)$ would be easily determined: if $|X|=1$ or $G[X]$ is connected, then trivially $ST(G, X)=0$; and, if $X = \{u,v\}$, then $ST(G, X)$ is equal to the number of vertices in any minimum path between $u$ and $v$ in $G$. 

We say that two distinct vertices $u,v \in V(G)$ are \emph{twins} in $G$ if they have the same neighborhood in $G$, \emph{i.e.} either $N_{G}(u) = N_{G}(v)$ or $N_{G}[u]=N_{G}[v]$. 
We prove in the next lemma that we can suppose without loss of generality that $G$ has no twins. Observe that the hypothesis involving $u$ and $v$ can be assumed without loss of generality. It is included in order to write the equation in a more concise way.


\begin{lemma}\label{lemma:twin}
Let $G$ be a connected graph, containing twin vertices $u$ and $v$, and $X\subseteq V(G)$ with $|X|\ge 3$. Also, suppose that $u\in X$ implies $v\in X$. Then,
\[ST(G,X) = ST(G-u,X-u).\]
\end{lemma}
\begin{proof}
First, let $S\subseteq V(G-u)\setminus X$ be a minimum set such that $S\cup (X-u)$ induces a connected subgraph of $G-u$. 
We want to prove that $S\cup X$ induces a connected subgraph of $G$, in which case we get $ST(G,X)\le ST(G-u,X-u)$. 
Suppose first that $v\in S$. 
Since $X\setminus\set{u} \neq \emptyset$, $v$ must have some neighbor $w\in S\cup X$ in $G-u$. 
Then, it follows from the hypothesis that $u$ and $v$ are twins in $G$ that $w$ is also a neighbor of $u$ in $G$. This implies that $S\cup X$ induces a connected subgraph of $G$. 
A similar argument can be applied when $v\in X$. 
Indeed, since $|X|\ge 3$, $X\setminus\set{u,v} \neq \emptyset$. 
Thus, $v$ must have some neighbor $w\in S\cup (X-u)$, which is also a neighbor of $u$ in $G$, and consequently $S\cup X$ induces a connected subgraph of $G$. 
Finally, if $v\notin S\cup X$, then by hypothesis we also know that $u\notin X$, in which case trivially $S\cup X$ induces a connected subgraph of $G$. 

Now, let $S\subseteq V(G)\setminus X$ be a minimum set such that $S\cup X$ induces a connected subgraph of $G$. 
Since $S$ is minimum, $|S \cap \{u,v\}|\le 1$. 
If $u \in S$, then, by the minimality of $S$, $v \not\in S$ and $(S \setminus \set{u}) \cup \set{v}$ witnesses $ST(G-u,X-u)\le ST(G,X)$. On the other hand, if $u\notin S$, then we trivially get that $S\cup (X-u)$ induces a connected subgraph of $G-u$, and therefore $ST(G-u,X-u)\le ST(G,X)$. 
\end{proof}

Based on Lemma~\ref{lemma:twin}, we assume from now on that the input graph $G$ has no twin vertices. 
Also, in the remainder of the text, $(T,\{T_u\}_{u\in V(G)})$ is a tree model of $G$. 
Moreover, given a node $t\in V(T)$, we denote by $V_t$ the set $\{u\in V(G)\colon t\in V(T_u)\}$. 
We say that $u\in V(G)$ is a \emph{leafy vertex} if $V(T_u) = \{\ell_u\}$ and $\ell_u$ is a leaf in $T$; denote by ${\cal L}$ the set of leafy vertices, and for every $u\in {\cal L}$, denote by $\ell_u$ the unique node in $T_u$. 
We also say that $(T,\{T_u\}_{u\in V(G)})$ is \emph{minimal} if there are no two adjacent nodes $t, t' \in V(T)$ such that $V_t \subseteq V_{t'}$. 
It is known that such a tree model can be computed in polynomial time~\cite{Gavril1978}. 
We prove in the following lemma that, for any minimal tree model $(T,\{T_u\}_{u\in V(G)})$, there is a one-to-one correspondence between the leaves of $T$ and the leafy vertices associated with $T$. 

\begin{lemma}\label{lemma:bijection_leafy_and_leaf}
    Let $G$ be a connected undirected path graph without twin vertices, and $(T,\{T_u\}_{u\in V(G)})$ be a minimal tree model for $G$. 
    Then, for every leaf $\ell$ of $T$, there exists a unique $u\in {\cal L}$ such that ${\cal L}\cap V_{\ell} = \{u\}$.  
\end{lemma}
\begin{proof}
    Since $G$ has no twin vertices, for every leaf $t$ of $T$, there exists at most one leafy vertex $u$ of $G$ associated with $T$ such that $\ell_{u} = t$. 
    On the other hand, suppose that there exists a leaf $t$ of $T$ such that there is no leafy vertex of $G$ associated with $T$ corresponding to $t$, i.e., for every leafy vertex $u$ of $G$ associated with $T$, we have that $\ell_{u} \neq t$. 
    Then, let $t'$ be the parent of $t$ in $T$. 
    One can readily verify that $V_{t} \subseteq V_{t'}$, contradicting the fact that we are on a minimal tree model. 
\end{proof}

A vertex $u$ is called \emph{simplicial} if $N_{G}(u)$ is a clique in $G$. Note that every leafy vertex is simplicial. 
Moreover, note that a simplicial vertex that is not in $X$ certainly is not contained in any minimum Steiner tree for $X$; therefore we can suppose that $X$ contains every simplicial vertex of $G$ and, in particular, ${\cal L}\subseteq X$. 
In the next lemma, we prove that we can suppose that every $x\in X$ is either leafy, or is such that $T_x$ contains no leaf of $T$.


\begin{lemma}\label{lemma:leafy_ornoleaf}
Let $G$ be a connected undirected path graph without twin vertices.  Also, let $(T,\{T_u\}_{u\in V(G)})$ be a minimal tree model of $G$, ${\cal L}$ be the set of leafy vertices associated with $T$, and let $X\subseteq V(G)$ be a set of terminals such that ${\cal L}\subseteq X\subseteq V(G)$. Suppose that $x\in X\setminus {\cal L}$ is such that $\ell\in V(T_x)$ for some leaf $\ell$ of $T$. Then, $ST(G,X) = ST(G-u,X-u)$, where ${\cal L}\cap V_{\ell} = \{u\}$.
\end{lemma}
\begin{proof}
By Lemma~\ref{lemma:bijection_leafy_and_leaf}, $G-u$ is the graph related to the tree model $T-\ell$. 
Suppose that there exists a set $S\subseteq V(G-u)\setminus X$ such that $S\cup (X-u)$ induces a connected subgraph of $G-u$.  Since $\ell\in T_u\cap T_x$, $ux\in E(G)$. 
Thus, $S\cup X$ induces a connected subgraph of $G$, and consequently $ST(G,X) \le ST(G-u,X-u)$. 
Conversely, suppose that there exists a set $S\subseteq V(G)\setminus X$ such that $S\cup X$ induces a connected subgraph of $G$.
Since $u$ is a simplicial vertex of $G$, we obtain that $S\cup (X-u)$ induces a connected subgraph of $G-u$, and therefore $ST(G,X) \ge  ST(G-u,X-u)$.
\end{proof}

In the proof, we modify a subset $S$ that gives a solution in order to ensure that the new set is a clique. The following lemma will help us do that.

\begin{lemma}\label{lem:replacing}
Let $G$ be a connected undirected path graph,  $(T,\{T_u\}_{u\in V(G)})$ be a tree model of $G$, $X\subseteq V(G)$ be a set of terminals, $S\subseteq V(G)\setminus X$ be a set such that $G[S\cup X]$ is connected, and let $u,v\in S$. 
If $y,z\in V(G)$ are such that $N_{G}(u)\cup N_{G}(v)\subseteq N_{G}(y)\cup N_{G}(z)$, then $S'\cup X$ is connected, where $S'=((S\setminus\{u,v\})\cup \{y,z\})\setminus X$. 
\end{lemma}
\begin{proof}
Suppose otherwise, and let $H,H'$ be distinct components of $G[S'\cup X]$. This means that there exist $w\in V(H)$ and $w'\in V(H')$ such that every path $P$ between $w$ and $w'$ goes through $u$ and/or $v$; but since $N_{G}(u)\cup N_{G}(v)\subseteq N_{G}(y)\cup N_{G}(z)$, it means that $u$ and/or $v$ can be replaced by $y$ and/or $z$.
\end{proof}

We are now ready to prove our theorem. In~\cite{K.90}, the author proves that deciding whether an undirected path graph has a dominating clique of size at most $k$ can be done in polynomial time. We make a polynomial reduction from \problemName{Connected Dominating Set} to \problemName{Dominating Clique}, thus getting the desired polynomial algorithm by the equivalence given in~\cite{WFP.85}. This and Theorem~\ref{thm:SteinerUP} give a dichotomy of both \problemName{Connected Dominating Set} and \problemName{Steiner Tree} in terms of the diameter of the input graph $G$. 
%
%
%
%
We observe that, in order to prove that \problemName{Dominating Clique} is polynomial-time solvable for \className{Undirected Path} graphs, it is used in~\cite{K.90} an  alternative notion of tree model called \emph{characteristic tree}, where the nodes of the model are the maximal cliques of the graph $G$ (see~\cite{Gavril1978,MW.86}).



\begin{theorem}
\problemName{Steiner Tree} and \problemName{Connected Dominating Set} can be solved in polynomial time when restricted to undirected path graphs of diameter at most~2.
\end{theorem}
\begin{proof}
In view of the reduction from \problemName{Connected Dominating Set} to \problemName{Steiner Tree} presented in~\cite{WFP.85}, it suffices to prove that \problemName{Steiner Tree} can be solved in polynomial time. 
Thus, let $G$ be a connected undirected path graph with diameter at most~2, $X \subseteq V(G)$ be a terminal set such that $\abs{X} \geq 3$, and let $\param$ be a positive integer. 


As usual, we consider a minimal tree model $(T,\{T_u\}_{u\in V(G)})$ of $G$. 
There is no loss of generality making these assumptions, since, as previously mentioned, it is known that such a tree model can be computed in polynomial time~\cite{Gavril1978}. 
Let ${\cal L}$ denote the set of leafy vertices associated with $T$, and assume that ${\cal L} \subseteq X$. 

In what follows, we prove that: $ST(G,X)\le \param$ if and only if there exists a clique $S\subseteq V(G)\setminus X$ in $G$ of size at most $\param$ that dominates ${\cal L}$. Our theorem follows since this is exactly what is computed in the algorithm presented in~\cite{K.90}.

For the sufficiency part of our claim, we just note that if $S$ is a clique of $G$ that dominates ${\cal L}$, then $\bigcup_{u\in S}V(T_u) = V(T)$, which means that in fact $S$ is a dominating clique of $G$ and therefore $S\cup X$ induces a connected subgraph of $G$. Because $|S|\le \param$, it follows that $ST(G,X)\le \param$.

Now, to prove necessity, suppose first that $ST(G,X)\le \param$, and let $S\subseteq V(G)\setminus X$ be a set such that $S\cup X$ induces a connected subgraph of $G$. 
Suppose that $S$ is minimum and that, among all such subsets of minimum cardinality, $S$ maximizes the number of edges in $E(G[S])$. 
In addition, by Lemma~\ref{lemma:leafy_ornoleaf}, we can suppose that there are no edges in $G$ between the vertices belonging to $X\setminus{\cal L}$ and the vertices belonging to ${\cal L}$. 
Moreover, note that ${\cal L}$ is an independent set. Thus, since $G[S\cup X]$ is connected, we get that every vertex in ${\cal L}$ must be adjacent to some vertex in $S$; in other words, $S$ dominates ${\cal L}$. 
Thus, it remains to prove that $S$ is a clique of $G$. 
Suppose for the sake of contradiction that there exist two distinct vertices $u,v\in S$ such that $uv\notin E(G)$.
We prove that one can find a pair $x,y$ of adjacent vertices in $G$ such that $N_{G}(u)\cup N_{G}(v)\subseteq N_{G}(x)\cup N_{G}(z)$. Then, based on Lemma~\ref{lem:replacing}, by letting $S'=((S\setminus \{u,v\})\cup \{y,z\})\setminus X$, we obtain a contradiction, since in this case either $|S'|<|S|$, or $|E(G[S'])|>|E(G[S])|$.

\input{figTree}

Let $t_u\in V(T_u)$ and $t_v\in V(T_v)$ be nodes whose distance from each other in $T$ is the smallest possible (observe Figure~\ref{fig:Tree}). 
Note that $t_u\neq t_v$ since $V(T_u)\cap V(T_v)=\emptyset$. Also, let $P_u^1,P_u^2$ be the two subpaths defined by $t_u$ in $T_u$, and define $P^1_v,P^2_v$ similarly. For each $i\in \{1,2\}$, let $\ell^i_u$ be the end vertex of $P^i_u$ different from $t_u$, if it exists; otherwise, let $\ell^i_u$ be equal to $t_u$. Define $\ell^1_v,\ell^2_v$ similarly. Note that, if $\ell^1_u\neq t_u$, then we can suppose that there exist $w^1_u\in V(G)$ such that $\ell^1_u\in T_{w^1_u}$, and $t\notin T_{w^1_u}$, where $t$ is the neighbor of $\ell^1_u$ in $P^1_u$ (otherwise, we could contract the edge $\ell^1_ut$ and still have a tree model of $G$). Define $w^2_u,w^1_v,w^2_v$ similarly. 
There are two possible cases to be considered. 
\begin{enumerate}[{Case} 1.]
    \item Suppose that all the vertices $w^1_u,w^2_u,w^1_v,w^2_v$ exist and are well-defined. 
    Since $G$ has diameter at most~2, there must exist vertices $y\in N_{G}(w^1_u)\cap N_{G}(w^1_v)$ and $z\in N_{G}(w^2_u)\cap N_{G}(w^2_v)$. 
    Clearly, the path between $t_u$ and $t_v$ in $T$ is contained in the paths $T_y$ and $T_z$ (which means that $y$ and $z$ are adjacent in $G$), and $V(P^1_u\cup P^1_v)\subseteq V(T_y)$, and $V(P^2_u\cup P^2_v)\subseteq V(T_z)$. 
    Thus, $N_{G}(u)\cup N_{G}(v)\subseteq N_{G}(y)\cup N_{G}(z)$, as desired. 

    \item Now, suppose that some of the vertices  $w^1_u,w^2_u,w^1_v,w^2_v$ do not exist or are not well-defined. 
    Note that, since $S\cup X$ induces a connected subgraph of $G$, there must exist a path in $G[S\cup X]$ between $u$ and $v$, which means that there must exist $w\in (S\cup X)\setminus\{u\}$ such that $t_u\in V(T_w)$.  This implies that at least one of the vertices $w^1_u,w^2_u$ is well-defined, as otherwise $V(T_u) = \{t_u\}\subseteq V(T_w)$ and we could just remove $u$ from $S$. 
    The same argument can be applied with respect to $w^1_v,w^2_v$. 
    Thus, suppose without loss of generality that $w^1_u,w^1_v$ are well-defined, and that $w^2_u$ is not well-defined (which means that $V(P^2_u) =\{t_u\}$). Pick $y$ as before, and note that $V(T_u)\subseteq V(T_y)$, and that $yv\in E(G)$ since $\ell^1_v\in V(T_y)\cap V(T_v)$. We can then apply Lemma~\ref{lem:replacing} to $\{u,v\}$ and $\{y,v\}$ since $N_{G}(u)\cup N_{G}(v)\subseteq N_{G}(y)\cup N_{G}(v)$.\qedhere
\end{enumerate}
\end{proof}

%% file: graphGConnDomSet.tex
\usetikzlibrary{arrows,positioning,fit,calc,shapes.symbols,decorations.pathreplacing}

\begin{figure}[th]\centering
	\begin{tikzpicture}[scale=1.6,sibling distance=4.0em,  level distance=2.5em,
		every node/.style = {align=center,text=black,inner xsep=0.7mm, inner ysep=0.7mm,font=\scriptsize,label distance=0.2mm}, clique/.style = {rounded corners, draw=black!60}]

	  	\node [clique,label=left:{\tiny $K$}] {$\cup_{\striple_{j}\in\sset}\set{\avertex_{j},\bvertex_{j},\cvertex_{j},\xvertex_{j}}$} [child anchor=north]
		    child { node[clique] (AB) {$\set{\avertex_{j},\bvertex_{j},\xvertex_{j},\yvertex_{j}}$}
		    	child { node[clique] (A) {$\set{\avertex_{j},\yvertex_{j},\zvertex_{j}^{1}}$} }
		    	child { node[clique] (B) {$\set{\bvertex_{j},\yvertex_{j},\zvertex_{j}^{2}}$} }
		    }
	        child { node[clique] (C) {$\set{\cvertex_{j},\xvertex_{j},\zvertex_{j}^{3}}$}}
			child { node[clique,label=below:{$\forall\, \pvertex_{i}\in\pset$}] (p) {$\begin{aligned}\set{\pvertex_{i}} \cup \\ \set{\avertex_{j} \setst \pvertex_{i} \in \striple_{j}} \end{aligned}$}}
			child { node[clique,label=below:{$\forall\, \qvertex_{i}\in\qset$}] (q) {$\begin{aligned}\set{\qvertex_{i}} \cup \\ \set{\bvertex_{j} \setst \qvertex_{i} \in \striple_{j}}\end{aligned}$}}
			child { node[clique,label=below:{$\forall\, \rvertex_{i}\in\rset$}] (r) {$\begin{aligned}\set{\rvertex_{i}} \cup \\ \set{\cvertex_{j} \setst \rvertex_{i} \in \striple_{j}}\end{aligned}$}}
	        ; 

	    \node[clique,thick,dotted,draw=gray,inner sep=1mm,label=above left:{$\forall\,\striple_{j}\in\sset$},fit=(AB) (A) (C)] (Vj) {};
	\end{tikzpicture}
	\caption{A tree model $\tuple{T,\set{\tree_{\uvertex}}_{\uvertex\in\gvset{\graph}}}$ associated with the graph $\graph$, constructed from a given instance $\instance=\tuple{\pset,\qset,\rset,\sset}$ of the \problemName{3D-Matching} problem.}\label{graphGConnDomSet}
\end{figure}


%% file: figVj.tex
\begin{figure}[thb]\centering
\begin{tikzpicture}[font=\scriptsize,label distance=0.1mm]
    \matrix (ps) [matrix of nodes, column sep=4mm, row sep=4mm,nodes={draw,circle,minimum size=0.1mm,inner sep=2pt}] {
      \node[label=above:{$\zvertex_{j}^{1}$}] (z1) {}; & \node[label=above:{$\yvertex_{j}$}] (y) {}; & \node[label=above:{$\zvertex_{j}^{2}$}] (z2) {}; \\
      \node[label=left:{$\avertex_{j}$}] (a) {};  &                & \node[label=right:{$\bvertex_{j}$}] (b) {}; \\
      \node[label=left:{$\xvertex_{j}$}] (x) {};  &                & \node[label=right:{$\cvertex_{j}$}] (c) {}; \\
                     & \node[label=below:{$\zvertex_{j}^{3}$}] (z3) {}; &                \\
    };
    \graph [use existing nodes] {
          z1 -- {a, y};
          y -- {a, b, z2, x};
          z2 -- {b};
          a -- {b, x, c};
          b -- {x, c};
          x -- {c,z3};
          c -- {z3};
        };
\end{tikzpicture}
\caption{Subgraph of $G$ induced by $\vset_{j}$, for  $\striple_{j}\in\sset$.}\label{fig:Vj}
\end{figure}
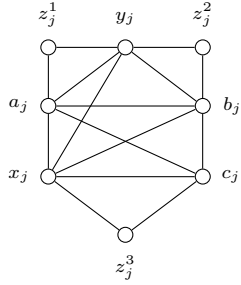

%% file: figTree.tex
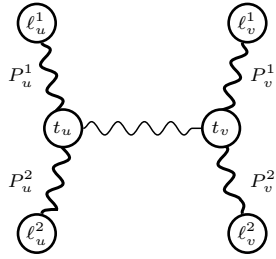
\begin{figure}[ht]\centering
\begin{tikzpicture}[scale=0.70]\scriptsize
  \pgfsetlinewidth{1pt}
  \tikzset{node/.style={ellipse, minimum size=0.5cm,draw, inner sep=1pt,font=\scriptsize}}
  \tikzset{snake it/.style={decorate, decoration=snake}}

\node[node] (ellu1) at (-2.5,2) {$\ell_u^1$};
\node[node] (ellu2) at (-2.5,-2) {$\ell_u^2$};
\draw[snake it, line width=0.04cm] (ellu1)--(-2,0)--(ellu2);
\node[node,fill=white] (tu) at (-2,0) {$t_u$};
\node at (-2.8,1) {$P_u^1$};
\node at (-2.8,-1) {$P_u^2$};

\node[node] (ellv1) at (1.5,2) {$\ell_v^1$};
\node[node] (ellv2) at (1.5,-2) {$\ell_v^2$};
\draw[snake it, line width=0.04cm] (ellv1)--(1,0)--(ellv2);
\node[node,fill=white] (tv) at (1,0) {$t_v$};
\node at (1.8,1) {$P_v^1$};
\node at (1.8,-1) {$P_v^2$};

\draw[snake it, line width=0.02cm] (tu)--(tv);

\end{tikzpicture}

\caption{The bold  lines represent the paths $T_u$ and $T_v$.}\label{fig:Tree}
\end{figure}

%% file: 05-Conclusion.tex
\section{Stubborn Puzzles 35 Years Later}
After 40 years, two open problems from [GJ] are still unsolved, namely: Open1 \problemName{Graph Isomorphism}, and Open8 Precedence constrained 3-processor schedule. In STOC 2016, László Babai announced that \problemName{Graph Isomorphism} could be solved in Quasipolynomial Time.
Only one {\COpen} entry from [OG] remains stubbornly open for 35 years: the complexity of \problemName{Chromatic Index} for \className{Planar} graphs. It is a puzzle to understand
why still today the \problemName{Chromatic Index} column has the majority of thirteen \COpen? entries, for instance \problemName{Chromatic Index} for \className{Cographs} is a long-standing open problem, as mentioned in~\cite{KR.03}.
We invite the reader to find a reference or a proof for the underlined [OG] entry in Table~\ref{table:first} corresponding to a ``private communication'' which would classify as polynomial \problemName{Steiner Tree} restricted to \className{Circle} graphs. 
Our proposed Table 2 leaves as open the parameterized complexity classification of \problemName{Partition into Cliques} for \className{Line} graphs.
We invite the reader to further study the eight $\XP$ entries, observing that five of them belong to our target \problemName{Steiner Tree} column. 
In particular, we highlight that, even though the closely related problem \problemName{Connected Dominating Set} is known to be $\FPT$ for \className{Claw-Free} graphs~\cite{Hermelin2019}, it is open whether \problemName{$\param$-Steiner Tree} is also $\FPT$ for \className{Line} and \className{Claw-Free} graphs.
Regarding the obtained second dichotomy for the \problemName{Steiner Tree} problem restricted to \className{Undirected Path} graphs, according to the diameter of the input graph, we should mention that \problemName{Connected Dominating Set} was proven to be \NP-complete and $\W[2]$-hard
even when restricted to \className{Split} graphs of diameter~2~\cite{Lokshtanov2013}.
 A straightforward modification of their proof leads to the \NP-completeness of \problemName{Steiner Tree} (and to the $\W[2]$-hardness of \problemName{$\param$-Steiner Tree}) when restricted to \className{Split} graphs of diameter $2$.